\documentclass[sn-mathphys,Numbered]{sn-jnl}

\usepackage{graphicx}
\usepackage{subfigure}
\usepackage{multirow}%
\usepackage{amsmath,amssymb,amsfonts}%
\usepackage{amsthm}%
\usepackage{mathrsfs}%
\usepackage[title]{appendix}%
\usepackage{xcolor}%
\usepackage{textcomp}%
\usepackage{manyfoot}%
\usepackage{booktabs}%
\usepackage{algpseudocode}%
\usepackage{listings}%
\usepackage[final]{changes}
\RequirePackage{algorithm,algpseudocode}

\newtheorem{thm}{Theorem}
\newtheorem{defin}{Definition}
\newtheorem{lem}{Lemma}
\newtheorem{assum}{Assumption}
\newtheorem{rem}{Remark}
\newtheorem{cor}{Corollary}

\newtheorem{prop}{Proposition}
\newtheorem{Ex}{Example}

\begin{document}

\title[Mixed membership distribution-free model]{Mixed membership distribution-free model}

\author*[1]{\fnm{Huan} \sur{Qing}}\email{qinghuan@u.nus.edu}
\author[2]{\fnm{Jingli} \sur{Wang}}\email{jlwang@nankai.edu.cn}
\affil[1]{\orgdiv{School of Economics and Finance}, \orgname{Chongqing University of Technology}, \city{Chongqing}, \postcode{400054}, \state{Chongqing}, \country{China}}

\affil[2]{\orgdiv{School of Statistics and Data Science, KLMDASR, LEBPS, and LPMC}, \orgname{Nankai University}, \city{Tianjin}, \postcode{300071}, \state{Tianjin}, \country{China}}


\abstract{We consider the problem of community detection in overlapping weighted networks, where nodes can belong to multiple communities and edge weights can be finite real numbers. To model such complex networks, we propose a general framework - the mixed membership distribution-free (MMDF) model. MMDF has no distribution constraints of edge weights and can be viewed as generalizations of some previous models, including the well-known mixed membership stochastic blockmodels. Especially, overlapping signed networks with latent community structures can also be generated from our model. We use an efficient spectral algorithm with a theoretical guarantee of convergence rate to estimate community memberships under the model. We also propose the fuzzy weighted modularity to evaluate the quality of community detection for overlapping weighted networks with positive and negative edge weights. We then provide a method to determine the number of communities for weighted networks by taking advantage of our fuzzy weighted modularity. Numerical simulations and real data applications are carried out to demonstrate the usefulness of our mixed membership distribution-free model and our fuzzy weighted modularity.}



\maketitle
\section{Introduction}\label{sec1}
For decades, the problem of community detection for networks has been actively studied in network science. The goal of community detection is to infer latent node's community information from the network, and community detection serves as a useful tool to learn network structure \cite{fortunato2010community,fortunato2016community,papadopoulos2012community}. Most of the networks that have been studied in literature are unweighted, i.e., the edges between vertices are either present or not \cite{newman2004analysis}. To solve the problem of community detection, researchers usually propose statistical models to model a network \cite{goldenberg2010a}. The classical Stochastic Blockmodel (SBM) \cite{SBM} models non-overlapping unweighted networks by assuming that the probability of an edge between two nodes depends on their respective communities and each node only belongs to one community. Recent developments of SBM can be found in \cite{abbe2017community}. In real-world networks, nodes may have an overlapping property and belong to multiple communities \cite{xie2013overlapping} while SBM can not model overlapping networks. The Mixed Membership Stochastic Blockmodel (MMSB) \cite{MMSB} extends SBM by allowing nodes to belong to multiple communities. Models proposed in \cite{DCSBM,OCCAM,MixedSCORE} extend SBM and MMSB by introducing node variation to model networks in which node degree varies. Based on SBM and its extensions, substantial works on algorithms, applications, and theoretical guarantees have been developed, to name a few, \cite{rohe2011spectral,choi2011stochastic,lei2015consistency,abbe2015community, SCORE,joseph2016impact,abbe2016exact,chen2018convexified,mao2020estimating,qing2023regularized}.

However, the above models are built for unweighted networks and they can not model weighted networks, where an edge weight can represent the strength of the connection between nodes. Weighted networks are ubiquitous in our daily life. For example, in the neural network of the Caenorhabditis elegans worm \cite{watts1998collective}, a link joins
two neurons if they are connected by either a synapse or a gap junction and the weight of a link is the number of synapses
and gap junctions \cite{opsahl2009clustering}, i.e., edge weights for the neural network are nonnegative integers; in the network of the 500 busiest commercial airports in the United States \cite{colizza2007reaction,opsahl2008prominence}, two airports are linked if a flight was scheduled between them in 2002 and the weight of a link is the number of available seats on the scheduled flights \cite{opsahl2009clustering}, i.e., edge weights for the US 500 airport network are nonnegative integers; in co-authorship networks \cite{liu2005co}, two authors are connected if they have co-authored at least one paper and the weight of a link is the number of papers they co-authored, i.e., edge weights for co-authorship network are nonnegative integers; in signed networks like the Gahuku-Gama subtribes network \cite{read1954cultures, yang2007community,kunegis2009slashdot,tang2016survey}, edge weights range in $\{1,-1\}$; in the Wikipedia conflict network, \cite{brandes2009network,kunegis2013konect}, a link represents a conflict between two users, link sign denotes positive and negative interactions, and the link weight denotes how strong the interaction is. For the Wikipedia conflict network, edge weights are real values.  To model non-overlapping weighted networks in which a node only belongs to one community, some models which can be viewed as SBM's extensions are proposed \cite{aicher2015learning,palowitch2018significance,xu2020optimal,ng2021weighted,qing2023DFM,BiDFMs} in recent years. However, these models can not model overlapping weighted networks in which nodes may belong to multiple communities. Though the multi-way blockmodels proposed in \cite{airoldi2013multi} can model mixed membership weighted networks (we also use mixed membership to denote overlapping occasionally), it has a strong requirement such that edge weights must be random variables generated from Normal distribution or Bernoulli distribution. This requirement makes the multi-way blockmodels fail to model the aforementioned weighted networks with nonnegative edge weights and signed networks. In this paper, we aim at closing this gap by building a general model for overlapping weighted networks.

The main contributions of this work include:

(1) We provide a general Mixed Membership Distribution-Free (MMDF for short) model for overlapping weighted networks in which a node can belong to multiple communities and an edge weight can be any real number. MMDF allows edge weights to follow any distribution as long as the expected adjacency matrix has a block structure. The classical mixed membership stochastic blockmodel is a sub-model of MMDF and overlapping signed networks with latent community structure can be generated from MMDF.

(2) We use a spectral algorithm to fit MMDF. We show that the proposed algorithm stably yields consistent community detection under MMDF. Especially, theoretical results when edge weights follow a specific distribution can be obtained immediately from our results.

(3) We provide fuzzy weighted modularity to evaluate the quality of mixed membership community detection for overlapping weighted networks. We then provide a method to determine the number of communities for overlapping weighted networks by increasing the number of communities until the fuzzy weighted modularity does not increase.

(4) We conduct extensive experiments to illustrate the advantages of MMDF and fuzzy weighted modularity.

\textbf{\textit{Notations.}}
We take the following general notations in this paper. Write $[m]:=\{1,2,\ldots,m\}$ for any positive integer $m$. For a vector $x$ and fixed $q>0$, $\|x\|_{q}$ denotes its $l_{q}$-norm, and we drop $q$ for $\|x\|_{q}$ when $q$ is 2. For a matrix $M$, $M'$ denotes the transpose of the matrix $M$, $\|M\|$ denotes the spectral norm, $\|M\|_{F}$ denotes the Frobenius norm, $\|M\|_{2\rightarrow\infty}$ denotes the maximum $l_{2}$-norm of all the rows of $M$, and $\|M\|_{\infty}:=\mathrm{max}_{i}\sum_{j}|M(i,j)|$ denotes the maximum absolute row sum of $M$. Let $\mathrm{rank}(M)$ denote the rank of matrix $M$. Let $\sigma_{i}(M)$ be the $i$-th largest singular value of matrix $M$, $\lambda_{i}(M)$ denote the $i$-th largest eigenvalue of the matrix $M$ ordered by the magnitude, and $\kappa(M)$ denote the condition number of $M$. $M(i,:)$ and $M(:,j)$ denote the $i$-th row and the $j$-th column of matrix $M$, respectively. $M(S_{r},:)$ and $M(:,S_{c})$ denote the rows and columns in the index sets $S_{r}$ and $S_{c}$ of matrix $M$, respectively. For any matrix $M$, we simply use $Y=\mathrm{max}(0, M)$ to represent $Y(i,j)=\mathrm{max}(0, M(i,j))$ for any $i,j$. $e_{i}$ is the indicator vector with a $1$ in entry $i$ and $0$ in all others.
\section{The Mixed Membership Distribution-Free Model}\label{intMMDF}
Consider an undirected weighted network $\mathcal{N}$ with $n$ nodes $\{1,2,\ldots,n\}$. Let $A\in \mathbb{R}^{n\times n}$ be the symmetric adjacency matrix of $\mathcal{N}$ such that $A(i,j)$ denotes the weight between node $i$ and node $j$ for $i,j\in[n]$. As a convention, we do not consider self-edges, so $A$'s diagonal elements are 0. $A$'s elements are allowed to be any finite real values instead of simple nonnegative values or 0 and 1. We assume all nodes in $\mathcal{N}$ belong to $K$ perceivable communities
\begin{align}\label{DefinCommunity}
\mathcal{C}^{(1)},\mathcal{C}^{(2)},\ldots,\mathcal{C}^{(K)}.
\end{align}
Unless specified, $K$ is assumed to be a known integer throughout this paper.
Since we consider mixed membership weighted networks in this paper, a node in $\mathcal{N}$ may belong to multiple communities with different weights. Let $\Pi\in \mathbb{R}^{n\times K}$ be the membership matrix such that for $i\in[n],k\in[K]$,
\begin{align}
&\mathrm{rank}(\Pi)=K, \Pi(i,k)\geq 0, \sum_{l=1}^{K}\Pi(i,l)=1,\label{DefinePI}\\
&\mathrm{Each~of~the ~}K\mathrm{~communities~has~at~least~one~pure~node}\label{pure},
\end{align}
where we call node $i$ `pure' if $\Pi(i,:)$ degenerates (i.e., one entry is 1, all others $K-1$ entries are 0) and `mixed' otherwise. In Equation (\ref{DefinePI}), since $\Pi(i,k)$ is the weight of node $i$ on community $\mathcal{C}^{(k)}$, $\|\Pi(i,:)\|_{1}=1$ means that $\Pi(i,:)$ is a $1\times K$ probability mass function (PMF) for node $i$. Equation (\ref{pure}) is important for the identifiability of our model. For convenience, call Equation (\ref{pure}) as pure nodes condition. For generating overlapping networks, the pure node condition is necessary for the model's identifiability, see models for overlapping unweighted networks considered in \cite{mao2020estimating, OCCAM,MixedSCORE}. Let $\mathcal{I}$ be the index of nodes corresponding to $K$ pure nodes, one from each community. Without loss of generality, let $\Pi(\mathcal{I},:)=I_{K}$, where $I_{K}$ is the $K\times K$ identity matrix.

Define a $K\times K$ connectivity matrix $P$ which satisfies
\begin{align}\label{definP}
 P=P',~\mathrm{rank}(P)=K,\mathrm{and~}\mathrm{max}_{k,l\in[K]}|P(k,l)|=1.
\end{align}
We'd emphasize that $P$ may have negative elements, the full rank requirement of $P$ is mainly for the identifiability of our model, and we set the maximum absolute value of $P$'s entries as 1 mainly for convenience. 

Let $\rho>0$ be a scaling parameter. For \emph{arbitrary distribution} $\mathcal{F}$ and all pairs of $(i,j)$ with $i,j\in[n]$, our model constructs the adjacency matrix of the undirected weighted network $\mathcal{N}$ such that $A(i,j)$ are independent random variables generated from $\mathcal{F}$ with expectation
\begin{align}\label{DefinOmega}
\mathbb{E}[A(i,j)]=\Omega(i,j), \mathrm{where~}\Omega:=\rho \Pi P\Pi'.
\end{align}
Call $\Omega$ the population adjacency matrix in this paper. Equation (\ref{DefinOmega}) means that  all elements of $A$ are independent random variables generated from an arbitrary distribution $\mathcal{F}$ with expectation $\Omega(i,j)$, without any prior knowledge on a specific distribution of $A(i,j)$ for $i,j\in[n]$. For comparison, mixed membership models considered in \cite{MMSB, OCCAM,MaoSVM,mao2020estimating,MixedSCORE} require all entries of $A$ are random variables generated from Bernoulli distribution with expectation $\Omega(i,j)$ since these models only model unweighted networks. Meanwhile, $\rho$'s range can vary for different distributions $\mathcal{F}$, see Examples \ref{NormalF}-\ref{SignedF} for detail.
\begin{defin}
(Model) Call Equations (\ref{DefinCommunity})-(\ref{DefinOmega}) the Mixed Membership Distribution-Free (MMDF) model and denote it by $MMDF_{n}(K, P, \Pi, \rho,\mathcal{F})$.
\end{defin}
The next proposition guarantees the identifiability of MMDF.
\begin{prop}\label{idMMDF}
(Identifiability). MMDF is identifiable: For eligible  $(P,\Pi)$ and $(\tilde{P}, \tilde{\Pi})$, if $\rho \Pi P\Pi'=\rho \tilde{\Pi}\tilde{P}\tilde{\Pi}'$, then $(\Pi,P)$ and $(\tilde{\Pi},\tilde{P})$ are identical up to a permutation of the $K$ communities.
\end{prop}
All proofs of proposition, lemmas, and theorems are provided in the Appendix. MMDF includes some previous models as special cases.
\begin{itemize}
  \item When $\mathcal{F}$ is a Bernoulli distribution, MMDF reduces to MMSB \cite{MMSB}.
  \item When all nodes are pure, MMDF reduces to the distribution-free model \cite{qing2023DFM}.
  \item When $\mathcal{F}$ is a Poisson distribution, all entries of $P$ are nonnegative, and $\Pi$ follows Dirichlet distribution, MMDF reduces to the weighted MMSB of \cite{dulac2020mixed}.
  \item When all nodes are pure and $\mathcal{F}$ is a Bernoulli distribution, MMDF reduces to SBM \cite{SBM}.
  \item When $K=1$ and $\mathcal{F}$ is Bernoulli distribution, MMDF degenerates to the Erd\"os-R\'enyi random graph \cite{erdos1960evolution}.
\end{itemize}
\begin{rem}
(Comparison to LFR benchmark networks) In \cite{lancichinetti2008benchmark}, the authors proposed LFR benchmark graphs for testing community detection algorithms on non-overlapping unweighted networks. In \cite{lancichinetti2009benchmarks}, the authors proposed generalizations of LFR benchmark networks for testing community detection methods on overlapping weighted graphs. \added{The advantage of the LFR benchmark network and its extension is that they allow the degree of nodes to follow the power law distributions. The disadvantage of the LFR benchmark network and its extension is that they require that each node in a community has a fraction $1-\mu$ of its links with the other nodes of its community and a fraction $\mu$ with nodes from the other communities, where the mixing parameter $\mu$ is fixed for all nodes \cite{lancichinetti2008benchmark}. For comparison, though SBM, MMSB, and our MMDF can not guarantee that the degrees of nodes follow the power law distribution, they can control the expectation of node degrees. Another advantage of SBM, MMSB, and our MMDF is, that they allow nodes from different communities to connect with different probabilities. Meanwhile, SBM, MMSB, and our MMDF have theoretical guarantees of consistency for algorithms fitting them.} \deleted{Though the LFR benchmark network and its extensions can allow the degrees of nodes to follow the power law distribution while SBM, MMSB, and our MMDF can not, they do not have any theoretical guarantees of consistency for algorithms fitting them while SBM, MMSB, and our MMDF have theoretical guarantees of consistency and researchers can further study the theoretical property of our MMDF just as SBM and MMSB which have been widely studied in recent years.} For example, the DFSP algorithm provided in the next section is designed to fit our MMDF model, and a theoretical guarantee of consistency for the DFSP algorithm is provided in Section \ref{ConsistencyMMDF}. We also further analyze the influences of the model parameter $\rho$ on DFSP's performance under different distributions $\mathcal{F}$ in Section \ref{ConsistencyMMDF}. Furthermore, though the extended LFR benchmark \cite{lancichinetti2009benchmarks} can generate overlapping weighted networks, it only models binary overlapping memberships and positive edge weights while our MMDF can model real-valued membership vectors and negative edge weights. In particular, the LFR benchmark network and its extensions can not model overlapping signed networks while our MMDF can.
\end{rem}
\section{Algorithm}
The goal of community detection under MMDF is to recover the membership matrix $\Pi$ with network $\mathcal{N}$'s adjacency matrix $A$ and the known number of communities $K$, where $A$ is generated from $MMDF_{n}(K, P,\Pi, \rho,\mathcal{F})$ for any distribution $\mathcal{F}$. To estimate $\Pi$ with given $A$ and $K$, we start by providing an intuition on designing a spectral algorithm to fit MMDF from the oracle case when $\Omega$ is known.

Since $\mathrm{rank}(P)=K$ and $\mathrm{rank}(\Pi)=K$, $\mathrm{rank}(\Omega)=K$ by basic algebra under $MMDF_{n}(K, P,\Pi,\rho,\mathcal{F})$. Let $\Omega=U\Lambda U'$ be the compact eigen-decomposition of $\Omega$ where $U\in\mathbb{R}^{n\times K}, \Lambda\in \mathbb{R}^{K\times K}$, and $U'U=I_{K}$. The following lemma functions similar to Lemma 2.1 of \cite{mao2020estimating} and guarantees the existence of Ideal Simplex (IS for short), and the form $U=\Pi B$ is called IS when $\Pi$ satisfies Equations (\ref{DefinePI})-(\ref{pure}).
\begin{lem}\label{IS}
(Ideal Simplex). Under $MMDF_{n}(K, P,\Pi, \rho,\mathcal{F})$, there exists an unique $K\times K$ matrix $B$ such that $U=\Pi B$ where $B=U(\mathcal{I},:)$.
\end{lem}
Given $\Omega$ and $K$, we can compute $U$ immediately by the top $K$ eigendecomposition of $\Omega$. Then, once we can obtain $U(\mathcal{I},:)$ from $U$, we can exactly recover $\Pi$ by $\Pi=UU^{-1}(\mathcal{I},:)$ since $U(\mathcal{I},:)\in\mathbb{R}^{K\times K}$ is a full rank matrix based on Lemma \ref{IS}.  As suggested by \cite{mao2020estimating}, for such IS, we can take advantage of the successive projection (SP) algorithm proposed in \cite{gillis2015semidefinite} (i.e., Algorithm \ref{alg:SP}) to $U$ with $K$ communities to exactly find the corner matrix $U(\mathcal{I},:)$.   For convenience, set $Z=UU^{-1}(\mathcal{I},:)$. Since $\Pi=Z$, we have $\Pi(i,:)=Z(i,:)\equiv\frac{Z(i,:)}{\|Z(i,:)\|_{1}}$ by the fact that $\|\Pi(i,:)\|_{1}=1$ for $i\in[n]$, where we write $\Pi(i,:)\equiv\frac{Z(i,:)}{\|Z(i,:)\|_{1}}$ mainly for the convenience to transfer the ideal algorithm given below to the real case.

The above analysis gives rise to the following algorithm called  Ideal DFSP (short for Distribution-Free SP algorithm) in the oracle case with known population adjacency matrix $\Omega$. Input $\Omega, K$. Output: $\Pi$.
\begin{itemize}
  \item Let $U$ be the top $K$ eigenvectors with unit-norm of $\Omega$.
  \item Run SP algorithm on all rows of $U$ with $K$ communities to obtain $\mathcal{I}$.
  \item Set $Z=UU^{-1}(\mathcal{I},:)$.
  \item Recover $\Pi$ by $\Pi(i,:)=Z(i,:)/\|Z(i,:)\|_{1}$ for $i\in[n]$.
\end{itemize}

Given $U$ and $K$, since the SP algorithm returns $\mathcal{I}$, we see that Ideal DFSP exactly returns $\Pi$, which supports the identifiability of MMDF in turn.

Next, we extend the ideal case to the real case. The community membership matrix $\Pi$ is unknown, and we aim at estimating it with given $(A,K)$ when $A$ is a random matrix generated from arbitrary distribution $\mathcal{F}$ under MMDF. Let $\tilde{A}=\hat{U}\hat{\Lambda}\hat{U}'$ be the top $K$ eigendecomposition of the adjacency matrix $A$ such that $\hat{U}\in \mathbb{R}^{n\times K}, \hat{\Lambda}\in \mathbb{R}^{K\times K},\hat{U}'\hat{U}=I_{K}$, and $\hat{\Lambda}$ contains the leading $K$ eigenvalues of $A$. Algorithm \ref{alg:DFSP} called DFSP is a natural extension of the Ideal DFSP to the real case, and DFSP is the SPACL algorithm without the prune step of \cite{mao2020estimating}, where we re-name it as DFSP to stress the distribution-free property of this algorithm.
\begin{algorithm}
\caption{DFSP}
\label{alg:DFSP}
\begin{algorithmic}[1]
\Require The adjacency matrix $A\in \mathbb{R}^{n\times n}$ and the number of communities $K$.
\Ensure The estimated $n\times K$ membership matrix $\hat{\Pi}$.
\State Compute $\tilde{A}=\hat{U}\hat{\Lambda}\hat{U}'$, the top $K$ eigendecomposition of $A$.
\State Run SP algorithm (i.e., Algorithm \ref{alg:SP}) on all rows of $\hat{U}$ with $K$ communities to obtain the estimated index set $\mathcal{\hat{I}}$ returned by SP.
\State Set $\hat{Z}=\mathrm{max}(0,\hat{U}\hat{U}^{-1}(\hat{\mathcal{I}},:))$.
\State Estimate $\Pi(i,:)$ by $\hat{\Pi}(i,:)=\hat{Z}(i,:)/\|\hat{Z}(i,:)\|_{1}, i\in[n]$.
\end{algorithmic}
\end{algorithm}

\begin{rem}\label{nonoverlapping}
DFSP can also obtain assignments to non-overlapping communities by setting $\hat{c}(i)=\mathrm{arg~max}_{1\leq k\leq K}\hat{\Pi}(i,k)$ for $i\in[n]$, where $\hat{c}(i)$ is the home base community that node $i$ belongs to.
\end{rem}

The time cost of DFSP mainly comes from the eigendecomposition step and SP step. The computational cost of top $K$ eigendecomposition is $O(Kn^{2})$ and the computational cost of SP is $O(nK^{2})$ \cite{MixedSCORE}. Because $K\ll n$ in this paper, as a result, the total computational complexity of DFSP is $O(Kn^{2})$. Results in  Section \ref{RealWeightedData} show that for a real-world weighted network with 13861 nodes, DFSP takes around 29 seconds to process a standard personal computer (Thinkpad X1 Carbon Gen 8) using MATLAB R2021b.
\section{Asymptotic Consistency}\label{ConsistencyMMDF}
In this section, we aim at proving that the estimated membership matrix $\hat{\Pi}$ concentrates around $\Pi$. Set $\tau=\mathrm{max}_{i,j\in[n]}|A(i,j)-\Omega(i,j)|$ and $\gamma=\frac{\mathrm{max}_{i,j\in[n]}\mathrm{Var}(A(i,j))}{\rho}$ where $\mathrm{Var}(A(i,j))$ denotes the variance of $A(i,j)$. $\tau$ and $\gamma$ are two parameters closely related to distribution $\mathcal{F}$. For different distribution $\mathcal{F}$, $\tau$ and $\gamma$ can be different, see Examples \ref{NormalF}-\ref{SignedF} for detail. For the theoretical study, we need the following assumption.
\begin{assum}\label{assumesparsity}
$\gamma\rho n\geq\tau^{2}\mathrm{log}(n)$.
\end{assum}
Assumption \ref{assumesparsity} means a lower bound requirement of $\gamma\rho n$. For different distribution $\mathcal{F}$, the exact form of Assumption \ref{assumesparsity} can be different because $\gamma$ relies on $\rho$, see Examples \ref{NormalF}-\ref{SignedF} for detail. The following theorem provides a theoretical upper bound on the $l_{1}$ errors of estimations for node memberships under MMDF.
\begin{thm}\label{Main}
(Error of DFSP) Under $MMDF_{n}(K,P,\Pi,\rho,\mathcal{F})$, let $\hat{\Pi}$ be obtained from Algorithm \ref{alg:DFSP}, when Assumption \ref{assumesparsity} holds, suppose $\sigma_{K}(\Omega)\geq C\sqrt{\gamma\rho n\mathrm{log}(n)}$ for some $C>0$, there exists a permutation matrix $\mathcal{P}\in\mathbb{R}^{K\times K}$ such that with probability at least $1-o(n^{-3})$, we have
\begin{align*}	\mathrm{max}_{i\in[n]}\|e'_{i}(\hat{\Pi}-\Pi\mathcal{P})\|_{1}=O(\frac{\kappa^{1.5}(\Pi'\Pi)\sqrt{\gamma n\mathrm{log}(n)}}{\sigma_{K}(P)\lambda_{K}(\Pi'\Pi)\sqrt{\rho}}).
\end{align*}
\end{thm}
Since our MMDF is distribution-free and $\mathcal{F}$ can be arbitrary distribution as long as Equation (\ref{DefinOmega}) holds, Theorem \ref{Main} provides a general theoretical upper bound of DFSP's error rate. Theorem \ref{Main} can be simplified by adding some conditions on $\lambda_{K}(\Pi'\Pi)$ and $K$, as shown by the following corollary.
\begin{cor}\label{AddConditions}
Under $MMDF_{n}(K,P,\Pi,\rho,\mathcal{F})$, when conditions of Theorem \ref{Main} hold, if we further suppose that $\lambda_{K}(\Pi'\Pi)=O(\frac{n}{K})$ and $K=O(1)$, with probability at least $1-o(n^{-3})$, we have
\begin{align*}
\mathrm{max}_{i\in[n]}\|e'_{i}(\hat{\Pi}-\Pi\mathcal{P})\|_{1}=O(\frac{1}{\sigma_{K}(P)}\sqrt{\frac{\gamma\mathrm{log}(n)}{\rho n}}).
\end{align*}
\end{cor}
In Corollary \ref{AddConditions}, the condition $\lambda_{K}(\Pi'\Pi)=O(\frac{n}{K})$ means that summations of nodes' weights in every community are in the same order, and $K=O(1)$ means that network $\mathcal{N}$ has a constant number of communities. The concise form of bound in Corollary \ref{AddConditions} is helpful for further analysis. By \cite{mao2020estimating}, we know that $\sigma_{K}(P)$ is a measure of the separation between communities and a larger $\sigma_{K}(P)$ means more well-separated communities. We are interested in the lower bound requirement on $\sigma_{K}(P)$ to make DFSP's error rate small. By Corollary \ref{AddConditions}, $\sigma_{K}(P)$ should shrink slower than $\sqrt{\frac{\gamma\mathrm{log}(n)}{\rho n}}$ for consistent estimation, i.e., $\sigma_{K}(P)\gg\sqrt{\frac{\gamma\mathrm{log}(n)}{\rho n}}$ should hold to make theoretical bound of error rate in Corollary \ref{AddConditions} go to zero as $n\rightarrow+\infty$. Meanwhile, Corollary \ref{AddConditions} says that DFSP stably yields consistent community detection under MMDF because the error bound in Corollary \ref{AddConditions} goes to zero as $n\rightarrow+\infty$.
\subsection{Examples}
The following examples provide $\gamma$'s upper bound, more specific forms of Assumption \ref{assumesparsity} and Theorem \ref{Main} for a specific distribution $\mathcal{F}$.
\begin{Ex}\label{NormalF}
When $A(i,j)\sim \mathrm{Normal}(\Omega(i,j),\sigma^{2}_{A})$ for $i,j\in[n]$. For this case, $\mathbb{E}[A(i,j)]=\Omega(i,j)$ holds by the property of Normal distribution, $P$ can have negative elements, $\rho$ ranges in $(0,+\infty)$ because the mean of Normal distribution can be any real values, $A\in\mathbb{R}^{n\times n}, \gamma=\frac{\sigma^{2}_{A}}{\rho}$, and $\tau$ is unknown. Setting $\gamma=\frac{\sigma^{2}_{A}}{\rho}$,  Assumption \ref{assumesparsity} becomes $\sigma^{2}_{A}n\geq\tau^{2}\mathrm{log}(n)$, a lower bound requirement on network size $n$. Setting $\gamma=\frac{\sigma^{2}_{A}}{\rho}$ in Theorem \ref{Main} obtains the theoretical upper bound of DFSP's $l_{1}$ error, and we see that increasing $\rho$ (or decreasing $\sigma^{2}_{A}$) decreases DFSP's error rate.
\end{Ex}
\begin{Ex}\label{BernoulliF}
When $A(i,j)\sim\mathrm{Bernoulli}(\Omega(i,j))$ for $i,j\in[n]$. $\mathbb{E}[A(i,j)]=\Omega(i,j)$ holds by the property of Bernoulli distribution, $P$ is an nonnegative matrix, $\rho$'s range is $(0,1]$ because $\rho P$ is a probability matrix, $A\in\{0,1\}^{n\times n}, \tau=1$, and $\gamma=\mathrm{Var}(A(i,j))/\rho=\Omega(i,j)(1-\Omega(i,j))/\rho\leq \Omega(i,j)/\rho\leq 1$. Setting $\gamma=1$ and $\tau=1$, Assumption \ref{assumesparsity} becomes $\rho n\geq\mathrm{log}(n)$. Setting $\gamma=1$ in Theorem \ref{Main} and we see that increasing $\rho$ decreases DFSP's error rate.
\end{Ex}
\begin{Ex}\label{PoissonF}
When $A(i,j)\sim \mathrm{Poisson}(\Omega(i,j))$ for $i,j\in[n]$. By the property of Poisson distribution, $\mathbb{E}[A(i,j)]=\Omega(i,j)$ holds, all entries of $P$ should be nonnegative, $\rho$'s range is $(0,+\infty)$ because the mean of Poisson distribution can be any positive number, $A$'s elements are nonnegative integers, $\tau$ is an unknown positive integer, and $\gamma=\mathrm{max}_{i,j\in[n]}\frac{\mathrm{Var}(A(i,j))}{\rho}=\mathrm{max}_{i,j\in[n]}\frac{\Omega(i,j)}{\rho}\leq 1$. Setting $\gamma=1$, Assumption \ref{assumesparsity} becomes $\rho n\geq\tau^{2}\mathrm{log}(n)$. Setting $\gamma=1$ in Theorem \ref{Main} and we see that increasing $\rho$ decreases DFSP's error rate.
\end{Ex}
\begin{Ex}\label{UniformF}
When $A(i,j)\sim\mathrm{Uniform}(0,2\Omega(i,j))$ for $i,j\in[n]$. For Uniform distribution, $\mathbb{E}[A(i,j)]=\frac{0+2\Omega(i,j)}{2}=\Omega(i,j)$ holds, $P$ is an nonnegative matrix, $\rho$'s range is $(0,+\infty)$, $A\in(0,2\rho)^{n\times n}$, $\tau$ is no larger than $2\rho$, and $\mathbb{E}[(A(i,j)-\Omega(i,j))^{2}]=\frac{4\Omega^{2}(i,j)}{12}\leq \frac{\rho^{2}}{3}$, i.e., $\gamma\leq\frac{\rho}{3}$. Setting $\gamma=\frac{\rho}{3}$, Assumption \ref{assumesparsity} becomes $\rho^{2}\geq\frac{3\tau^{2}\mathrm{log}(n)}{n}$. Setting $\gamma=\frac{\rho}{3}$ in Theorem \ref{Main}, we see that $\rho$ vanishes in the bound of the error rate, i.e., increasing $\rho$ does not influence DFSP's performance.
\end{Ex}
\begin{Ex}\label{SignedF}
MMDF can generate overlapping signed network by setting $\mathbb{P}(A(i,j)=1)=\frac{1+\Omega(i,j)}{2}$ and $\mathbb{P}(A(i,j)=-1)=\frac{1-\Omega(i,j)}{2}$ for $i,j\in[n]$. For this case, $\mathbb{E}[A(i,j)]=\Omega(i,j)$ holds, $P$ can have negative entries, $\rho$'s range is $(0,1)$ because $-1\leq \Omega(i,j)\leq1$, $A\in\{1,-1\}^{n\times n}, \tau\leq2$, and $\mathrm{Var}(A(i,j))=1-\Omega^{2}(i,j)\leq1$, i.e., $\gamma\leq \frac{1}{\rho}$. Setting $\gamma=\frac{1}{\rho}$ and $\tau=2$, Assumption \ref{assumesparsity} becomes $n\geq4\mathrm{log}(n)$. Setting $\gamma=\frac{1}{\rho}$ in Theorem \ref{Main} and we see that increasing $\rho$ decreases the error rate.
\end{Ex}
More than the above distributions, the distribution-free property of MMDF allows $\mathcal{F}$ to be any other distribution as long as Equation (\ref{DefinOmega}) holds. For example, $\mathcal{F}$ can be Binomial, Double exponential, Exponential, Gamma, and Laplace distributions in \url{http://www.stat.rice.edu/~dobelman/courses/texts/distributions.c&b.pdf}. Details on the probability mass function or probability density function on distributions discussed in this paper can also be found in the above URL link. Generally speaking, the distribution-free property guarantees the generality of our model MMDF, the DFSP algorithm, and our theoretical results.

\subsection{Missing edge}
From Examples \ref{NormalF}, \ref{UniformF}, and \ref{SignedF}, we find that $A(i,j)$ is almost always nonzero for $i\neq j$, which is impractical for real-world large-scale networks in which many nodes have no connections \cite{lei2015consistency}. Similar to \cite{xu2020optimal,qing2023DFM}, an edge with weight 0 is deemed as a missing edge. We generate missing edges for overlapping undirected weighted networks in the following way.

Let $\mathcal{A}\in\{0,1\}^{n\times n}$ be a symmetric and connected adjacency matrix of an undirected unweighted network. To model real-world large-scale overlapping undirected weighted networks with missing edges, for $i,j\in[n]$, we update $A(i,j)$ by $A(i,j)\mathcal{A}(i,j)$, where $A$ is generated from our model MMDF.  $\mathcal{A}$ can be generated from any models such as the Erd\"os-R\'enyi random graph $G(n,p)$ \cite{erdos1960evolution}, SBM, and MMSB as long as these models can generate undirected unweighted networks.

In particular, when $\mathcal{A}$ is generated from the Erd\"os-R\'enyi random graph $G(n,p)$ \cite{erdos1960evolution}, $\mathbb{P}(\mathcal{A}(i,j)=1)=p$ and $\mathbb{P}(\mathcal{A}(i,j)=0)=1-p$ for $i,j\in[n]$, i.e., increasing $p$ decreases the number of missing edges in $A$. We call $p$ sparsity parameter because it controls network sparsity. To make $\mathcal{A}$ be connected with high probability, we need $p\geq \frac{\mathrm{log}(n)}{n}$ \cite{abbe2017community}.

\section{Estimation of the Number of Communities}
In the DFSP algorithm, the number of communities $K$ should be known in advance, which is usually impractical for real-world networks. Here, we introduce fuzzy weighted modularity, then we combine it with DFSP to estimate $K$ for overlapping weighted networks.

Recall that $A$ considered in this paper can have negative edge weights, we define our fuzzy weighted modularity by considering both positive and negative edge weights in $A$. Let $A^{+}=\mathrm{max}(0,A)$ and $A^{-}=\mathrm{max}(0,-A)$, where $A^{+}$ and $A^{-}$ are two $n\times n$ symmetric matrices with nonnegative elements and $A(i,j)=A^{+}(i,j)-A^{-}(i,j)$ for $i,j\in[n]$. Let $d^{+}$ and $d^{-}$ be two  $n\times 1$ vectors such that  $d^{+}(i)=\sum_{j=1}^{n}A^{+}(i,j)$ and $d^{-}(i)=\sum_{j=1}^{n}A^{-}(i,j)$ for $i\in[n]$. Let $m^{+}=\sum_{i=1}^{n}d^{+}(i)/2$ and $m^{-}=\sum_{i=1}^{n}d^{-}(i)/2$.

For arbitrary community detection method $\mathcal{M}$, without confusion, let $\hat{\Pi}$ be an $n\times k$ estimated membership matrix returned by running method $\mathcal{M}$ on $A$ with $k$ communities, where all entries of $\hat{\Pi}$ are nonnegative and $\sum_{l}^{k}\hat{\Pi}(i,l)=1$ for $i\in[n]$. Based on the estimated membership matrix $\hat{\Pi}$, modularity $Q^{+}$  for $A$'s positive elements and modularity $Q^{-}$ for $A$'s negative elements are defined as
\begin{align*}
&Q^{+}=\frac{1}{2m^{+}}\sum_{i=1}^{n}\sum_{j=1}^{n}(A^{+}(i,j)-\frac{d^{+}(i)d^{+}(j)}{2m^{+}})\hat{\Pi}(i,:)\hat{\Pi}'(j,:)1_{m^{+}>0},\\ &Q^{-}=\frac{1}{2m^{-}}\sum_{i=1}^{n}\sum_{j=1}^{n}(A^{-}(i,j)-\frac{d^{-}(i)d^{-}(j)}{2m^{-}})\hat{\Pi}(i,:)\hat{\Pi}'(j,:)1_{m^{-}>0},
\end{align*}
where $1_{m^{+}>0}$ and $1_{m^{-}>0}$ are indicator functions such that $1_{m^{+}>0}=1$ if $m^{+}>0$ and 0 otherwise, $1_{m^{-}>0}=1$ if $m^{-}>0$ and 0 otherwise. We define our fuzzy weighted modularity as
\begin{align}\label{Modularity}
Q_{\mathcal{M}}(k)=\frac{2m^{+}}{2m^{+}+2m^{-}}Q^{+}-\frac{2m^{-}}{2m^{+}+2m^{-}}Q^{-}.
\end{align}
Our fuzzy weighted modularity computed using Equation~(\ref{Modularity}) measures the quality of overlapping community partition. Similar to the Newman-Girvan modularity \cite{newman2004finding,newman2006modularity}, a larger fuzzy weighted modularity $Q_{\mathcal{M}}(k)$ indicates a better estimation of community membership. Meanwhile, our fuzzy weighted modularity includes some previous modularity as special cases.
\begin{itemize}
  \item When all nodes in $\hat{\Pi}$ are pure, $A$ has both positive and negative elements (i.e., $m^{+}>0$ and $m^{-}>0$), our modularity reduces to the modularity for signed networks provided in Equation (17) of \cite{gomez2009analysis}.
  \item When all elements of $A$ are nonnegative (i.e., $m^{-}=0$), our modularity reduces to the fuzzy modularity provided in Equation (14) of \cite{nepusz2008fuzzy}.
  \item When all nodes in $\hat{\Pi}$ are pure and all elements of $A$ are nonnegative, our modularity reduces to the popular Newman-Girvan modularity \cite{newman2004finding,newman2006modularity}.
\end{itemize}
To determine the number of communities, we follow the strategy provided in \cite{nepusz2008fuzzy}. In detail, we iteratively increase $k$ and choose the one maximizing our fuzzy weighted modularity computed via Equation~(\ref{Modularity}) using method $\mathcal{M}$.
\begin{rem}
When all nodes are pure and $\Pi$, the ground truth of memberships, is known, Normalized Mutual Information (NMI) \cite{strehl2002cluster,danon2005comparing,bagrow2008evaluating} and Adjusted Rand Index (ARI) \cite{hubert1985comparing,vinh2009information} are two commonly used similarity metrics of community detection to measure the performance of a community detection algorithm. However, NMI and ARI require the true memberships to be known while $\Pi$ for most real-world networks is unknown. For such cases, like the famous  Newman-Girvan modularity \cite{newman2004finding,newman2006modularity}, our fuzzy weighted modularity can be used to measure the quality of overlapping community detection for weighted networks even with negative weights, i.e., our fuzzy weighted modularity can be a measure of overlapping weighted networks without real community memberships.
\end{rem}
\section{Experimental Results}
This section conducts extensive experiments to demonstrate that DFSP is effective for mixed membership community detection and our fuzzy weighted modularity is capable of the estimation of the number of communities for mixed membership weighted networks generated from our MMDF model. We conducted all experiments on a standard personal computer (Thinkpad X1 Carbon Gen 8) using MATLAB R2021b. First, we introduce comparison algorithms for each task. Next, evaluation metrics are introduced. Finally, we compare DFSP and our method for determining $K$ with their respective comparison algorithms on synthetic and real-world networks.
\subsection{Comparison algorithms}
For the task of community detection, we compare DFSP with the following algorithms.
\begin{itemize}
  \item \textbf{GeoNMF} \cite{GeoNMF} adapts nonnegative matrix factorization to infer community memberships for networks generated from MMSB for overlapping unweighted networks.
  \item \textbf{SVM-cone-DCMMSB (SVM-cD)} \cite{MaoSVM} and \textbf{OCCAM} \cite{OCCAM} are two spectral algorithms for detecting community memberships for networks generated from the overlapping continuous community assignment model for overlapping unweighted networks \cite{OCCAM}.
\end{itemize}
\begin{rem}
GeoNMF and OCCAM may fail to output when $A$ has negative entries. If this happens, to make them work, we make all elements of $A$ positive by adding a sufficiently large constant for GeoNMF and OCCAM.
\end{rem}
For the task of determining the number of communities,  we let $\mathcal{M}$ be our DFSP algorithm and call our strategy for determining $K$ as KDFSP, and we compare KDFSP with methods as follows.
\begin{itemize}
  \item \textbf{NB} and \textbf{BHac} \cite{le2022estimating} are two model-free approaches to estimate $K$. NB estimates $K$ based on the non-backtracking matrix and BHac is designed based on the Bethe Hessian matrix.
\end{itemize}
\subsection{Evaluation metrics}
For the task of mixed membership community detection, to evaluate the performance of different community detection approaches, different evaluation criteria are adopted according to the fact that whether the ground truth membership matrix $\Pi$ is known.

\textbf{Metrics for networks with ground truth.} For this case, we consider the following metrics to evaluate the performance of mixed membership community detection methods.
\begin{itemize}
  \item \textbf{Hamming error} measures the $l_{1}$ difference between the true membership matrix and the estimated membership matrix:
      \begin{align*}
\mathrm{Hamming~error=}\mathrm{min}_{\mathcal{P}\in\{ K\times K\mathrm{~permutation~matrix}\}}\frac{1}{n}\|\hat{\Pi}-\Pi\mathcal{P}\|_{1}.
\end{align*}
Hamming error ranges in $[0,1]$, and a smaller Hamming error indicates a better estimation of community membership. Note that our theoretical result in Theorem \ref{Main} is measured by Hamming error.
  \item \textbf{Relative error} measures the $l_{2}$ difference between $\Pi$ and $\hat{\Pi}$:
      \begin{align*}
\mathrm{Relative~error}=\mathrm{min}_{\mathcal{P}\in\{ K\times K\mathrm{~permutation~matrix}\}}\frac{\|\hat{\Pi}-\Pi\mathcal{P}\|_{F}}{\|\Pi\|_{F}}.
\end{align*}
Relative error is nonnegative, and it is the smaller the better.
\end{itemize}
We do not use metrics like NMI \cite{strehl2002cluster,danon2005comparing,bagrow2008evaluating}, ARI \cite{hubert1985comparing,vinh2009information}, and overlapping NMI \cite{lancichinetti2009detecting} because they require binary overlapping membership vectors \cite{MaoSVM} while $\Pi$ considered in this paper is not binary unless all nodes are pure.

\textbf{Metrics for networks without ground truth.} For this case, we use our fuzzy weighted modularity in Equation (\ref{Modularity}) to measure the quality of detected communities.

\textbf{Accuracy rate.} For the task of determining the number of communities, similar to \cite{le2022estimating}, we use Accuracy rate to measure the performance of KDFSP and its competitors in our simulation studies, where Accuracy rate is the fraction of times a method correctly estimates $K$.
\subsection{Simulations}\label{secSimulation}
For all simulations, unless specified, set $n=200, K=3$, and $n_{0}=40$, where $n_{0}$ denotes the number of pure nodes for each community. Let all mixed nodes have four different memberships $(0.4, 0.4, 0.2), (0.4, 0.2, 0.4), (0.2, 0.4, 0.4)$ and $(1/3,1/3,1/3)$, each with $\frac{n-3n_{0}}{4}$ number of nodes. The connectivity matrix $P$ and scaling parameter $\rho$ are set independently for each experiment.  Meanwhile, in all numerical studies, the only criteria for choosing $P$ is, $P$ should obey Equation (\ref{definP}), and $P$'s entries should be positive or negative relying on distribution $\mathcal{F}$. Each simulation experiment contains the following steps:

(a) Set $\Omega=\rho\Pi P\Pi'$.

(b) Let $A(i,j)$ be a random number generated from distribution $\mathcal{F}$ with expectation $\Omega(i,j)$ for $1\leq i<j\leq n$, set $A(j,i)=A(i,j)$ to make $A$ symmetric, and let $A$'s diagonal elements be zero since we do not consider self-edges.

(c) Apply DFSP and its competitors to $A$ with $K$ communities. Record Hamming error and Relative error.

(d) Apply KDFSP and its competitors to $A$. Record the estimated number of communities.

(e) Repeat (b)-(d) 100 times, and report the averaged Hamming error and Relative error, and the Accuracy rate.
\begin{rem}
For simplicity, we do not consider missing edges in our simulation study. Actually, similar to numerical results in \cite{qing2023DFM}, DFSP performs better as sparsity parameter $p$ increases when missing edges are generated from the Erd\"os-R\'enyi random graph $G(n,p)$.
\end{rem}
\subsubsection{Normal distribution}
Set $\mathcal{F}$ as a Normal distribution such that $A(i,j)\sim\mathrm{Normal}(\Omega(i,j),\sigma^{2}_{A})$ for some $\sigma^{2}_{A}$. Set $P$ as
\[P_{1}=\begin{bmatrix}
    1&-0.2&-0.3\\
    -0.2&0.9&0.3\\
    -0.3&0.3&0.9\\
\end{bmatrix}.\]
Let $\sigma^{2}_{A}=2$, and $\rho$ range in $\{5,10,\ldots,100\}$. The results are displayed in Panels (a)-(c) of Figure \ref{SIM}. For the task of community detection, we see that DFSP performs better as $\rho$ increases and this is consistent with our analysis in Example \ref{NormalF}. DFSP and GeoNMF perform similarly and both methods significantly outperform SVM-cD and OCCAM. For the task of inferring $K$, KDFSP performs slightly poorer than NB and BHac, and all methods perform better when $\rho$ increases.
\subsubsection{Bernoulli distribution}
Set $\mathcal{F}$ as a Bernoulli distribution such that $A(i,j)\sim\mathrm{Bernoulli}(\Omega(i,j))$. Set $P$ as
\[P_{2}=\begin{bmatrix}
    1&0.2&0.3\\
    0.2&0.9&0.3\\
    0.3&0.3&0.9\\
\end{bmatrix}.\]
Let $\rho$ range in $\{0.05,0.1,\ldots,1\}$. The results are shown in Panels (d)-(f) of Figure \ref{SIM}. For the task of community detection, we see that DFSP's error rates decrease when $\rho$ increases and this is consistent with our analysis in Example \ref{BernoulliF}. DFSP, GeoNMF, and SVM-cD enjoy similar performances and they perform better than OCCAM. For the task of inferring $K$, KDFSP outperforms its competitors, and all methods perform better when $\rho$ increases.
\subsubsection{Poisson distribution}
Set $\mathcal{F}$ as a Poisson distribution such that $A(i,j)\sim\mathrm{Poisson}(\Omega(i,j))$. Set $P$ as $P_{2}$. Let $\rho$ range in $\{0.2,0.4,\ldots,4\}$. Panels (g)-(i) of Figure \ref{SIM} show the results. For the task of community detection, we see that DFSP performs better when $\rho$ increases, which is consistent with our analysis in Example \ref{PoissonF}. Meanwhile, DFSP, GeoNMF, and SVM-cD have similar performances and they significantly outperform OCCAM. For the task of inferring $K$, KDFSP, NB, and BHac enjoy similar performances and all methods perform better when $\rho$ increases.
\subsubsection{Uniform distribution}
Set $\mathcal{F}$ as an Uniform distribution such that $A(i,j)\sim\mathrm{Uniform}(\Omega(i,j))$. Set $P$ as $P_{2}$. Let $\rho$ range in $\{1,2,\ldots,20\}$. Panels (j)-(l) of Figure \ref{SIM} show the results. For the task of community detection, we see that DFSP's error rates have no significant change when $\rho$ increases and this is consistent with our analysis in Example \ref{UniformF}. Meanwhile, DFSP and GeoNMF perform similarly, and they significantly outperform SVM-cD and OCCAM. For the task of determining $K$, KDFSP estimates $K$ correctly while it is interesting to see that NB and BHac fail to infer $K$ when $\rho$ increases.
\subsubsection{Signed network}
For a signed network when $\mathbb{P}(A(i,j)=1)=\frac{1+\Omega(i,j)}{2}$ and $\mathbb{P}(A(i,j)=-1)=\frac{1-\Omega(i,j)}{2}$, let $n=800$, each community have $n_{0}=200$ pure nodes, and $P$ be $P_{2}$. Let $\rho$ range in $\{0.05,0.1,\ldots,1\}$. Panels (m)-(o) of Figure \ref{SIM} display the results. For the task of community detection, we see that all methods enjoy similar behaviors and they perform better when $\rho$ increases, and this is consistent with our analysis in Example \ref{SignedF}. For the task of estimating $K$, KDFSP performs better when $\rho$ increases while NB and BHac fail to infer $K$.
\begin{figure}
\centering
\resizebox{\columnwidth}{!}{
\subfigure[Normal distribution]{\includegraphics[width=0.33\textwidth]{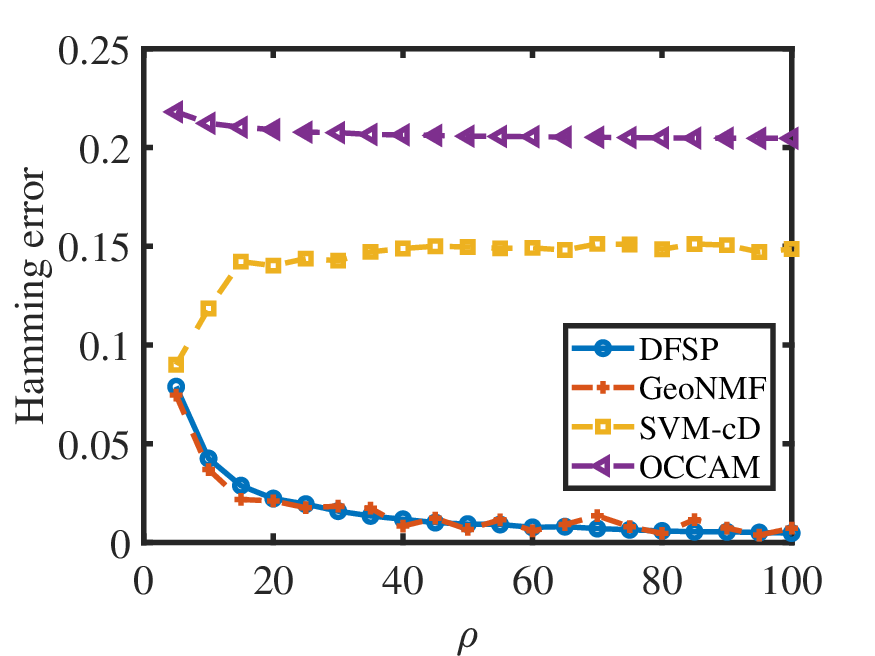}}
\subfigure[Normal distribution]{\includegraphics[width=0.33\textwidth]{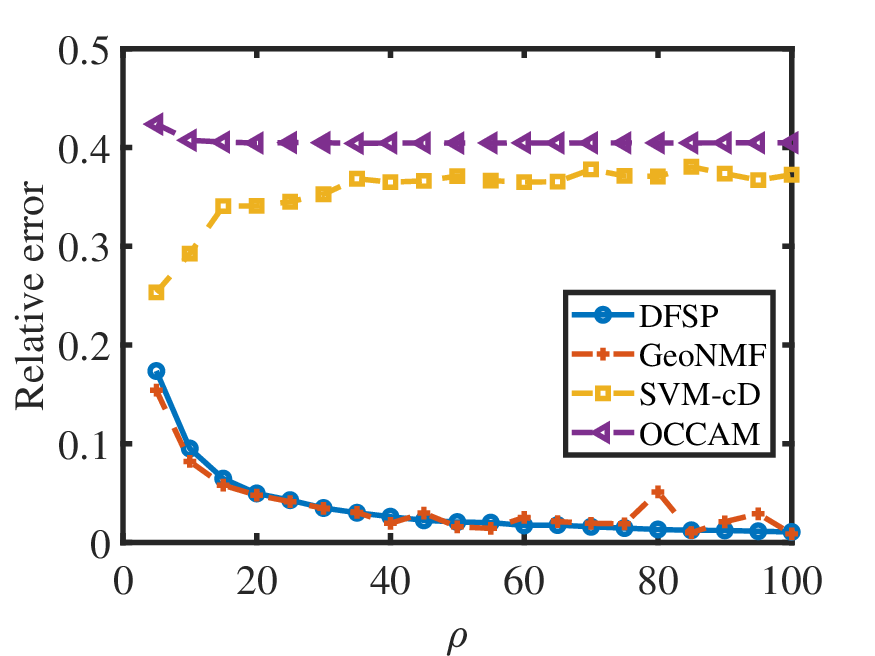}}
\subfigure[Normal distribution]{\includegraphics[width=0.33\textwidth]{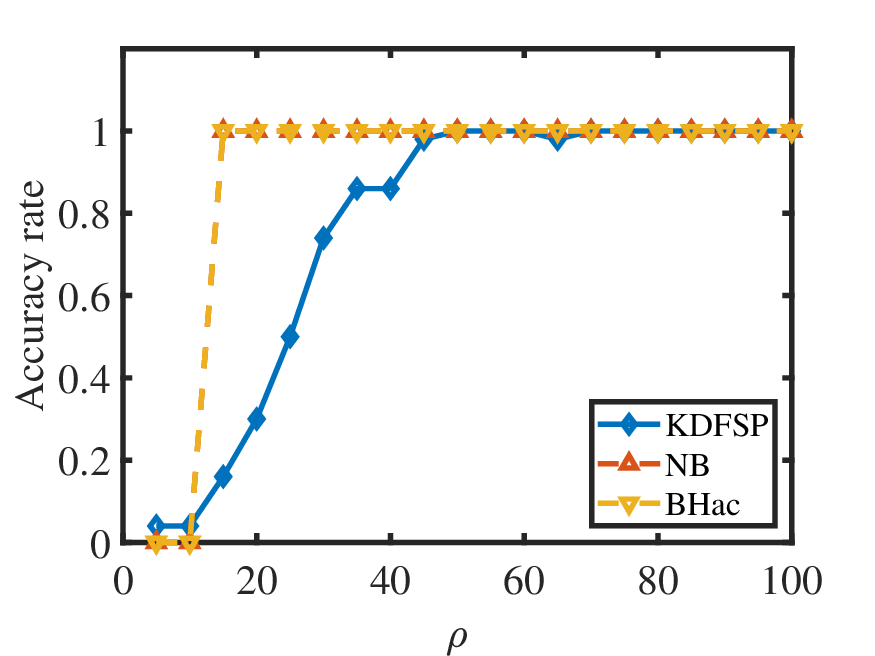}}
}
\resizebox{\columnwidth}{!}{
\subfigure[Bernoulli distribution]{\includegraphics[width=0.33\textwidth]{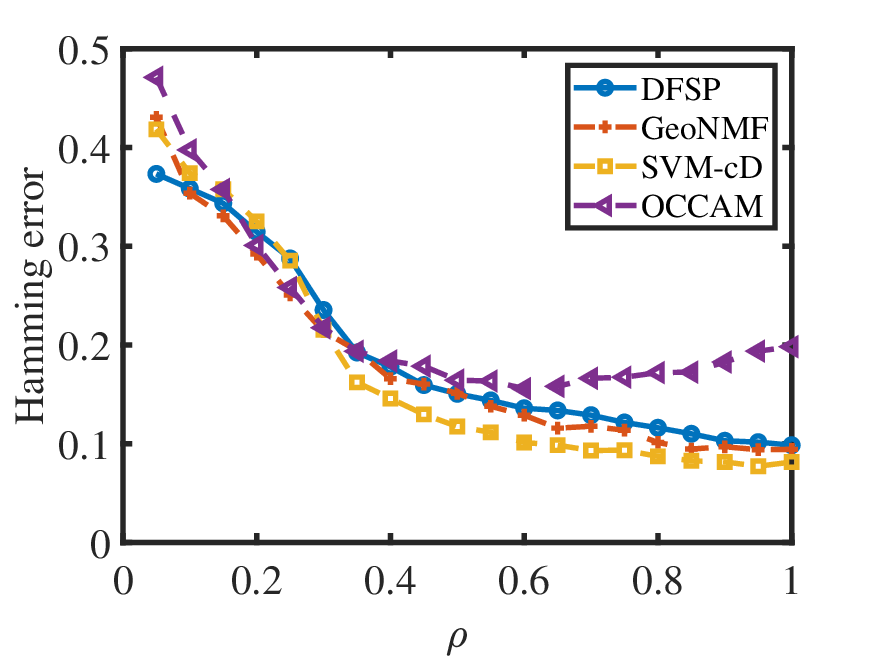}}
\subfigure[Bernoulli distribution]{\includegraphics[width=0.33\textwidth]{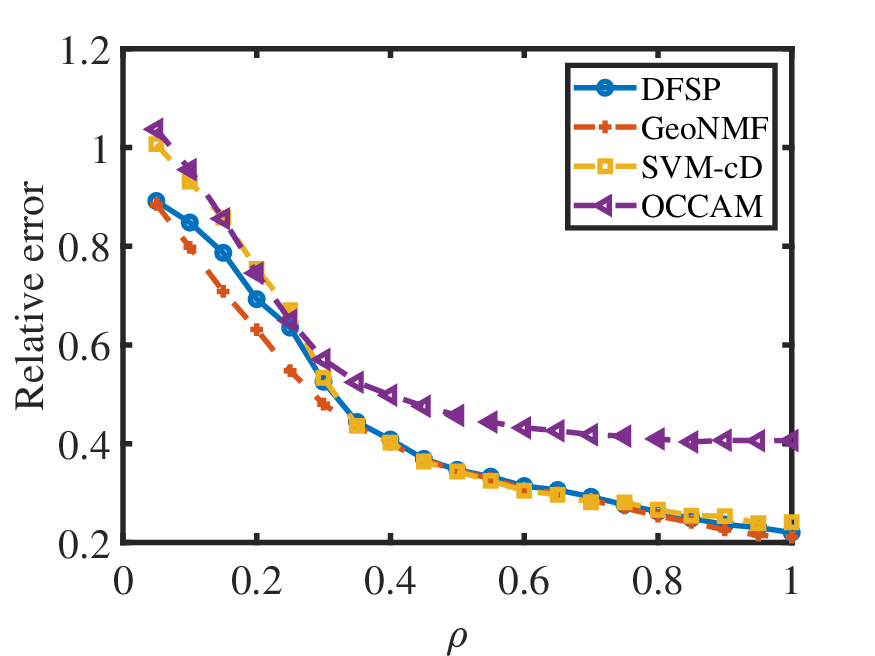}}
\subfigure[Bernoulli distribution]{\includegraphics[width=0.33\textwidth]{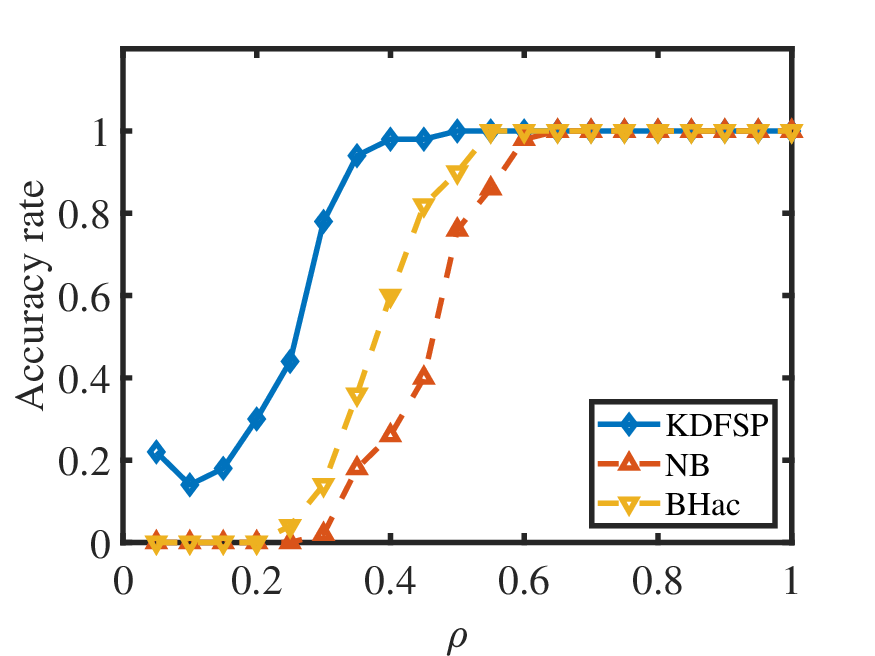}}
}
\resizebox{\columnwidth}{!}{
\subfigure[Poisson distribution]{\includegraphics[width=0.33\textwidth]{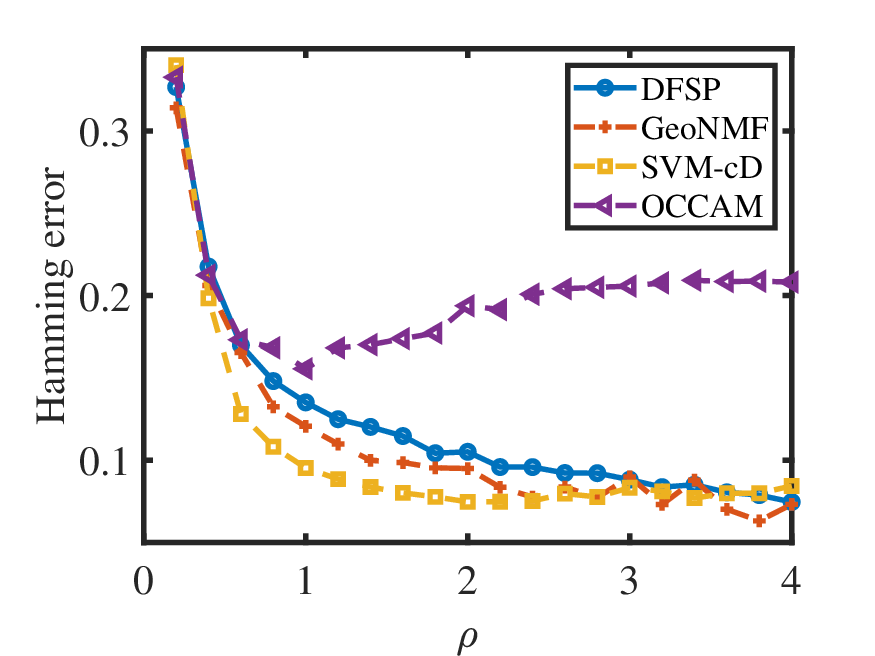}}
\subfigure[Poisson distribution]{\includegraphics[width=0.33\textwidth]{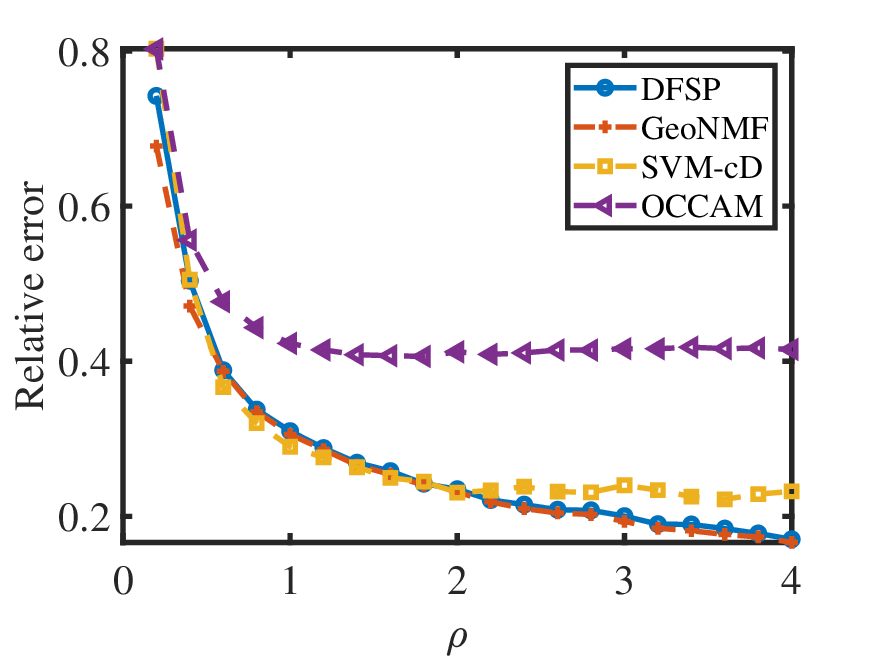}}
\subfigure[Poisson distribution]{\includegraphics[width=0.33\textwidth]{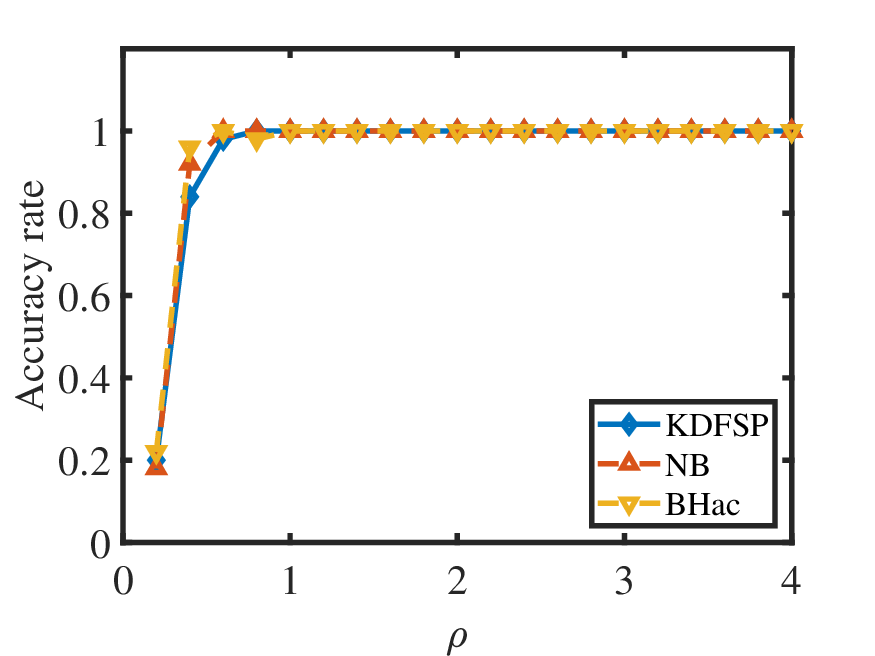}}
}
\resizebox{\columnwidth}{!}{
\subfigure[Uniform distribution]{\includegraphics[width=0.33\textwidth]{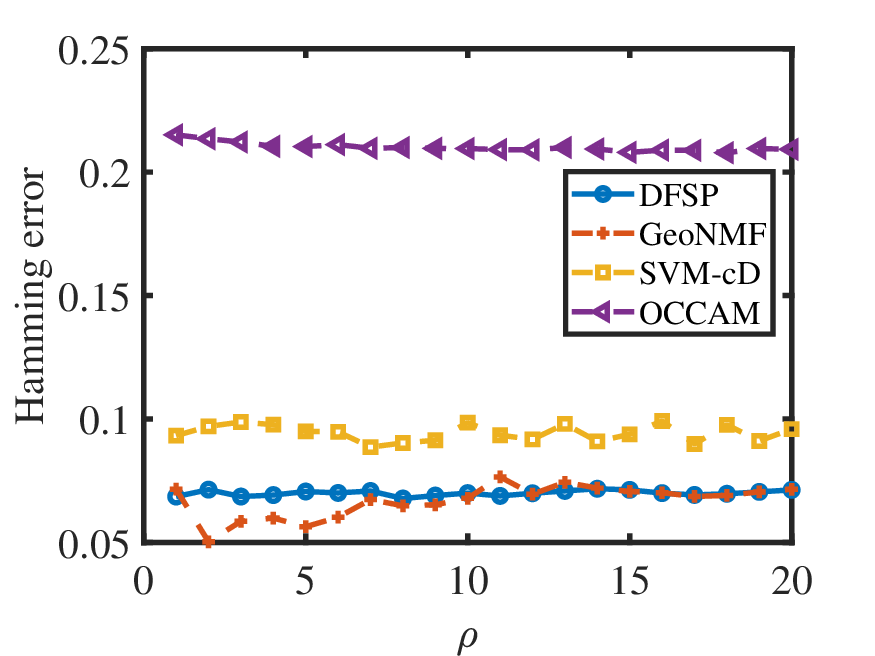}}
\subfigure[Uniform distribution]{\includegraphics[width=0.33\textwidth]{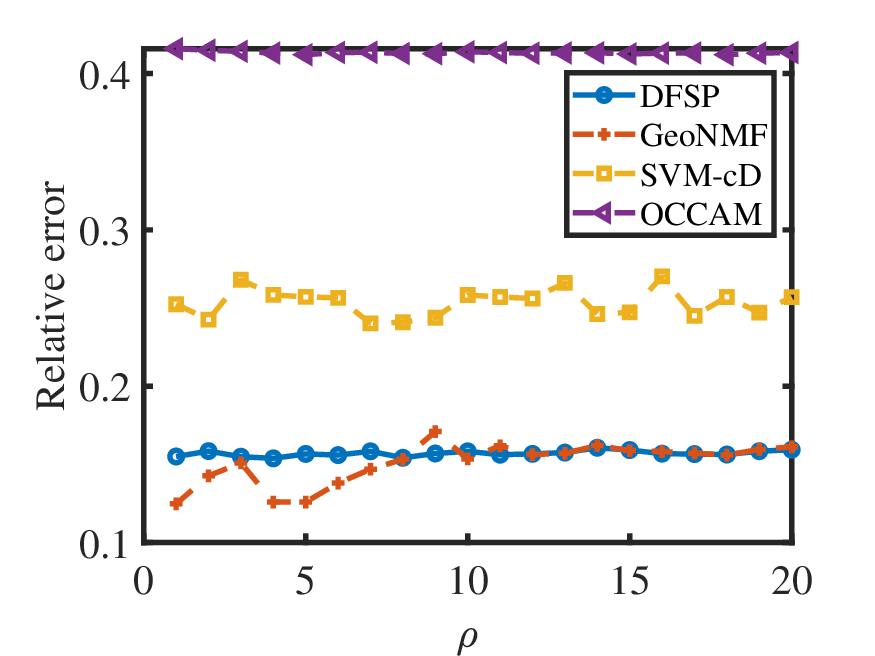}}
\subfigure[Uniform distribution]{\includegraphics[width=0.33\textwidth]{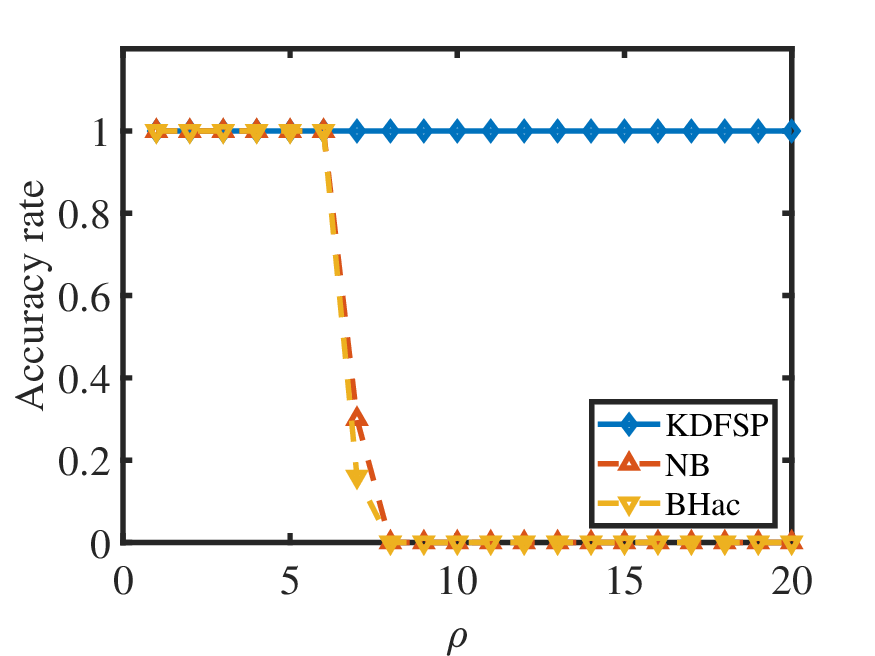}}
}
\resizebox{\columnwidth}{!}{
\subfigure[Signed network]{\includegraphics[width=0.33\textwidth]{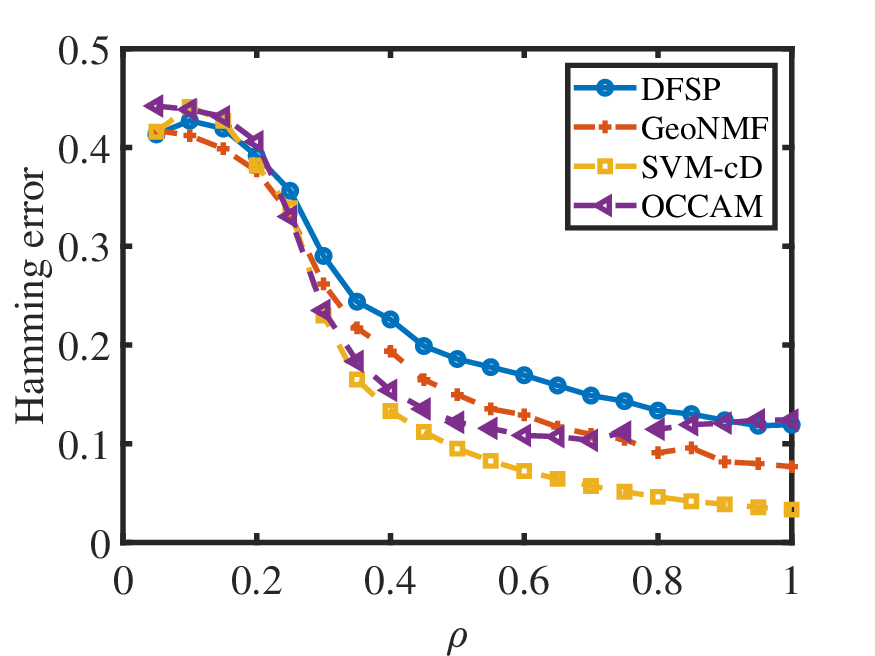}}
\subfigure[Signed network]{\includegraphics[width=0.33\textwidth]{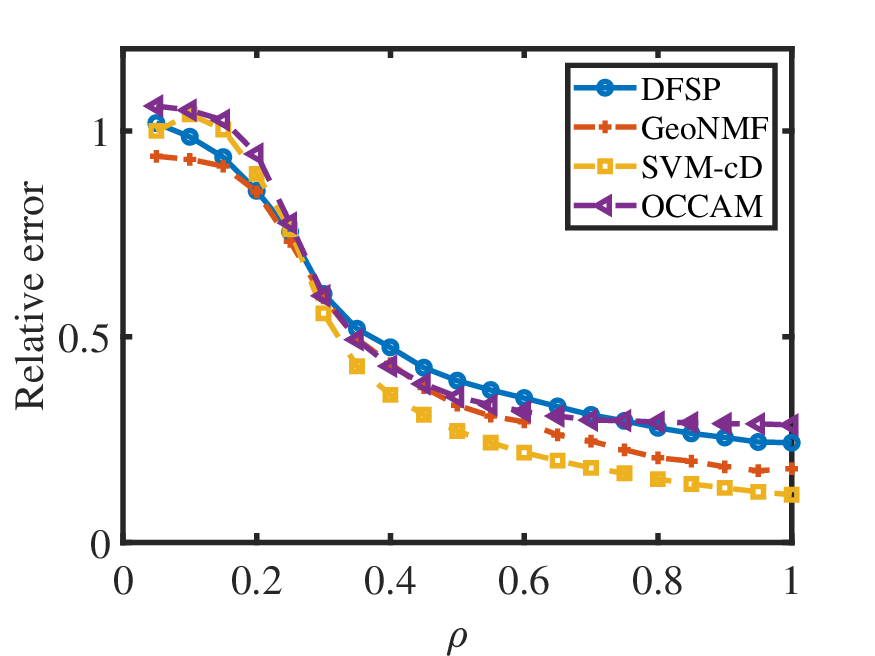}}
\subfigure[Signed network]{\includegraphics[width=0.33\textwidth]{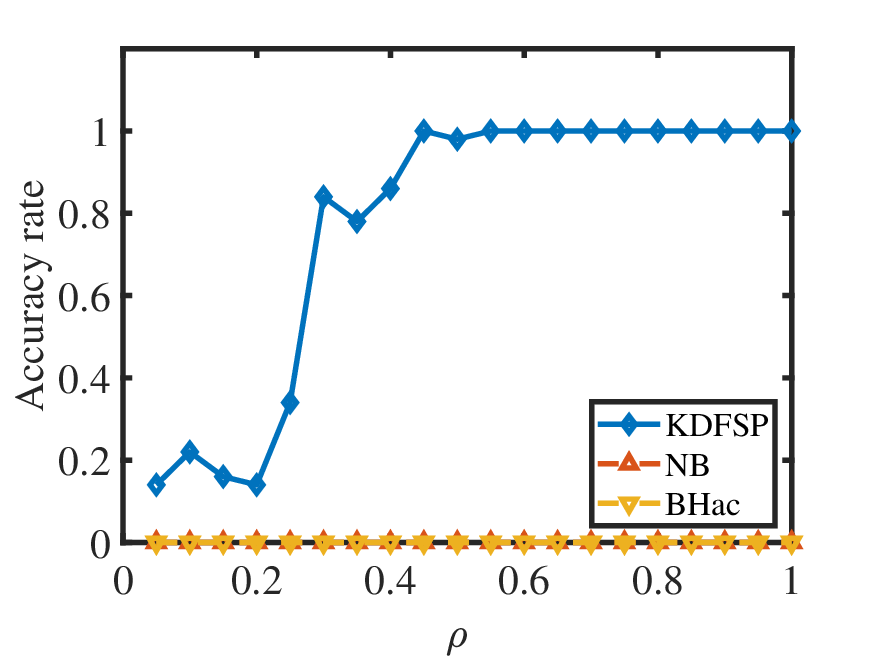}}
}
\caption{Numerical results for simulations.}
\label{SIM} 
\end{figure}
\subsection{Application to real-world networks}\label{RealWeightedData}
We use nine real-world networks to demonstrate the effectiveness of DFSP and KDFSP. Table \ref{realdata} summarizes basic information for these networks. For visualization, Figure \ref{AReal} displays adjacency matrices of the first three weighted networks.
\begin{table}[htp!]
	\tiny
\caption{Basic information and summarized statistics of real-world networks studied in this paper.}\label{realdata}
	\begin{tabular}{@{}lllllllllll@{}}
\toprule
\textbf{Dataset}&\textbf{Source}&\textbf{Node meaning}&\textbf{Edge meaning}&\textbf{Weighted?}&\textbf{True $\Pi$}&\boldmath{$n$}&\boldmath{$K$}\\
\midrule
Gahuku-Gama subtribes&\cite{read1954cultures}&Tribe&Friendship &Yes&Known&16&3\\
Karate-club-weighted&\cite{zachary1977information}&Member&Tie&Yes&Known&34&2\\
Slovene Parliamentary Party&\cite{ferligoj1996analysis}&Party&Relation&Yes&Unknown&10&2\\
Train bombing&\cite{hayes2006connecting}&Terrorist&Contact&Yes&Unknown&64&Unknown\\
Les Mis\'erables&\cite{knuth1993stanford}&Character&Co-occurence&Yes&Unknown&77&Unknown\\
US Top-500 Airport Network&\cite{colizza2007reaction}&Airport&\#Seats&Yes&Unknown&500&Unknown\\
Political blogs&\cite{Polblogs1}&Blog&Hyperlink&No&Known&1222&2\\
US airports&\cite{opsahl2011anchorage}&Airport&\#Flights&Yes&Unknown&1572&Unknown\\
cond-mat-1999&\cite{newman2001structure}&Scientist&Coauthorship&Yes&Unknown&13861&Unknown\\
\botrule
\end{tabular}
\end{table}

In Table \ref{realdata}, for networks with known memberships or $K$, their ground truth and $K$ are suggested by the original authors or data curators. For the Gahuku-Gama subtribes network, it can be downloaded from \url{http://konect.cc/networks/ucidata-gama/} and its node labels are shown in Figure 9 (b) \cite{yang2007community}. For the Karate-club-weighted network, it can be downloaded from \url{http://vlado.fmf.uni-lj.si/pub/networks/data/ucinet/ucidata.htm#kazalo} and its true node labels can be downloaded from \url{http://websites.umich.edu/~mejn/netdata/}. For the Slovene Parliamentary Party network, it can be downloaded from \url{http://vlado.fmf.uni-lj.si/pub/networks/data/soc/Samo/Stranke94.htm}.  For Train bombing, Les Mis\'erables, and US airports, they can be downloaded from \url{http://konect.cc/networks/} (see also \cite{kunegis2013konect}). The original US airports network has 1574 nodes and it is directed. We make it undirected by letting the weight of an edge be the summation of the number of flights between two airports. We then remove two airports that have no connections with any other airport. For US Top-500 Airport Network, it can be downloaded from \url{https://toreopsahl.com/datasets/#online_social_network}. For Political blogs, its adjacency matrix and true node labels can be downloaded from \url{http://zke.fas.harvard.edu/software/SCOREplus/Matlab/datasets/}. For the Condensed matter collaborations 1999 (Cond-mat-1999 for short) data, it can be downloaded from \url{http://websites.umich.edu/~mejn/netdata/}. Cond-mat-1999 has 16726 nodes and only 13861 nodes fall in the largest connected component which is the one we focus on in this paper.

\begin{figure}
\centering
\resizebox{\columnwidth}{!}{
\subfigure[Gahuku-Gama subtribes]{\includegraphics[width=0.35\textwidth]{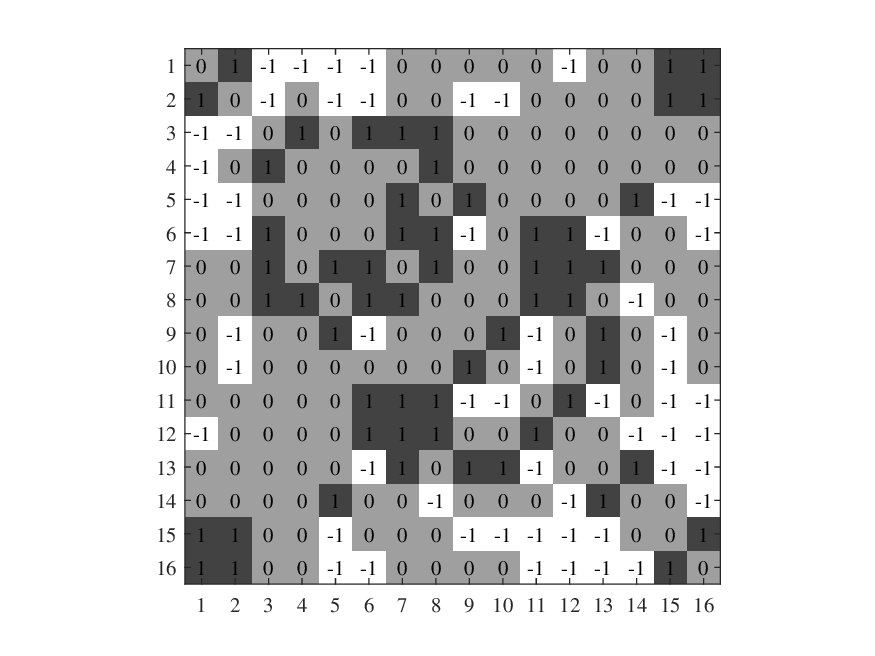}}
\subfigure[Karate-club-weighted]{\includegraphics[width=0.35\textwidth]{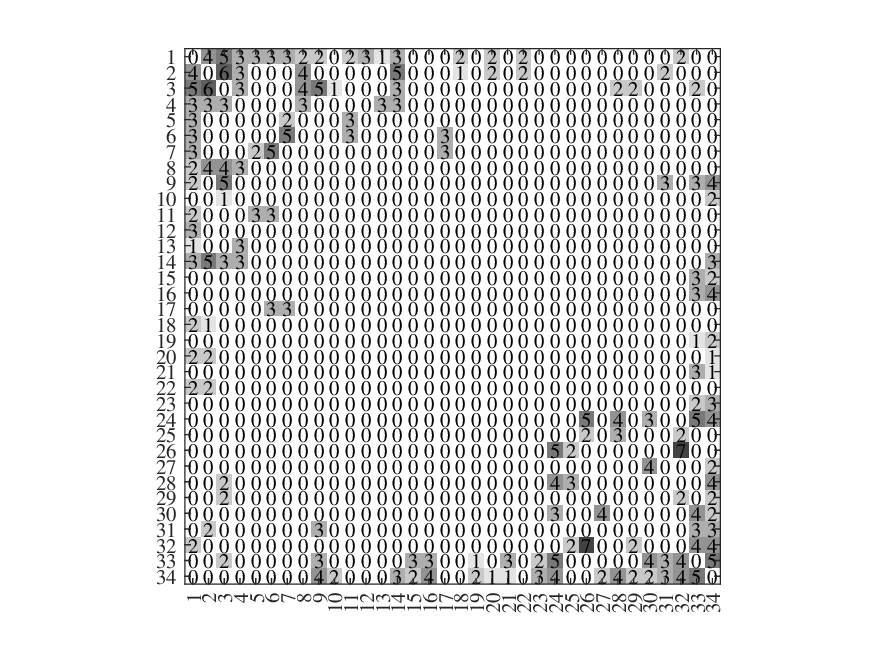}}
\subfigure[Slovene Parliamentary Party]{\includegraphics[width=0.35\textwidth]{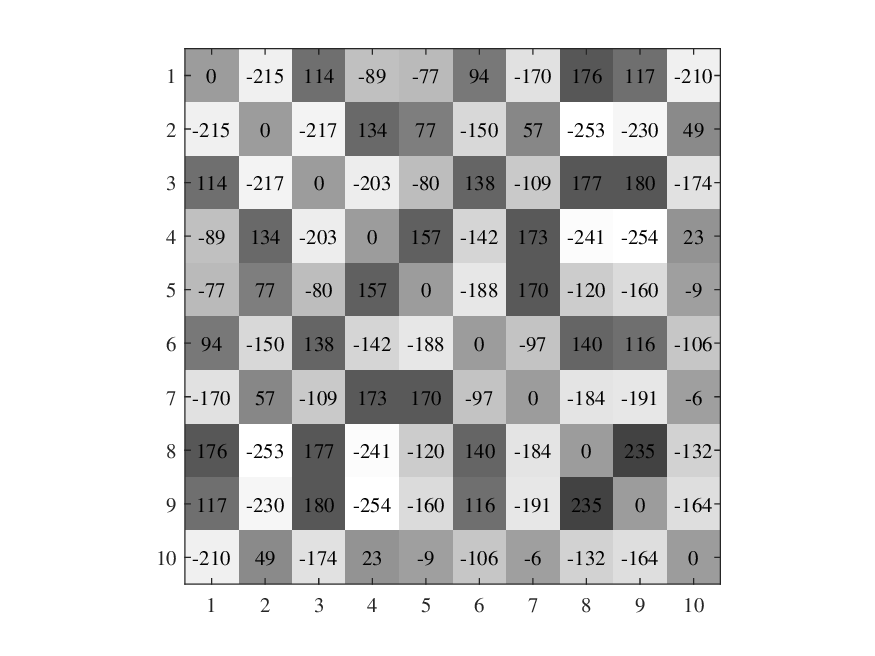}}
}
\caption{Adjacency matrices of Gahuku-Gama subtribes, Karate-club-weighted, and Slovene Parliamentary Party.}
\label{AReal} 
\end{figure}

Let $\hat{c}$ be the home base community vector computed by Remark \ref{nonoverlapping}. By comparing $\hat{c}$ with the true labels, for the Gahuku-Gama subtribes network, we find that DFSP misclusters 0 nodes out of 16; for the Karate-club-weighted network, DFSP misclusters 0 nodes out of 34; for the Political blogs network, DFSP misclusters 64 nodes out of 1222.

Table \ref{realdataEstimatedK} records the estimated number of communities of KDFSP and its competitors for real-world networks used in this paper. The results show that, for networks with known $K$, our KDFSP correctly determines the number of communities for these networks while NB and BHac fail to determine the correct $K$. For the first four networks with unknown $K$ in Table \ref{realdataEstimatedK}, our KDFSP determines their $K$ as 2 while $K$ inferred by NB and BHac is larger.For Cond-mat-1999, KDFSP chooses its $K$ as 38, a number much smaller than those determined by NB and BHac. Figure \ref{Qdfsp} shows the fuzzy weighted modularity $Q$ by the DFSP approach for different choices of the number of clusters. For each network, the number of clusters $k$ maximizing the fuzzy weighted modularity can be found clearly in Figure \ref{Qdfsp}.
\begin{table}[h!]
\footnotesize
	\centering
	\caption{Comparison of the estimated number of communities for real-world networks in Table \ref{realdata}.}
	\label{realdataEstimatedK}
	\begin{tabular}{cccccccccccc}
\toprule
\textbf{Dataset}&\textbf{True} \boldmath{$K$}&\textbf{KDFSP}&\textbf{NB}&\textbf{BHac}\\
\midrule
Gahuku-Gama subtribes&3&3&1&13\\
Karate-club-weighted&2&2&4&4\\
Slovene Parliamentary Party&2&2&N/A&N/A\\
Train bombing&Unknown&2&3&4\\
Les Mis\'erables&Unknown&2&6&7\\
US Top-500 Airport Network&Unknown&2&147&158\\
Political blogs&2&2&7&8\\
US airports&Unknown&2&100&137\\
Cond-mat-1999&Unknown&38&449&489\\
\botrule
\end{tabular}
\end{table}

From now on, we use the number of communities determined by KDFSP in Table \ref{realdataEstimatedK} for each data to estimate community memberships. We compare the fuzzy weighted modularity of DFSP and its competitors, and the results are displayed in Table \ref{Qalgorithms}. We see that DFSP returns larger fuzzy weighted modularity than its competitors except for the Karate-club-weighted network. Meanwhile, according to the fuzzy weighted modularity of DFSP in Table \ref{Qalgorithms}, we also find that Gahuku-Gama subtribes, Karate-club-weighted, Slovene Parliamentary Party, Les Mis\'erables, and Political blogs have a more clear community structure than Train bombing, US Top-500 Airport Network, US airports, and Cond-mat-1999 for their larger fuzzy weighted modularity. Furthermore, the running times of DFSP, GeoNMF, SVM-cD, and OCCAM for the Cond-mat-1999 network are 29.06 seconds, 32.33 seconds, 90.63 seconds, and 300 seconds, respectively. Hence, DFSP runs faster than its competitors.
\begin{figure}
\centering
\normalsize
\subfigure[Gahuku-Gama subtribes]{\includegraphics[width=0.323\textwidth]{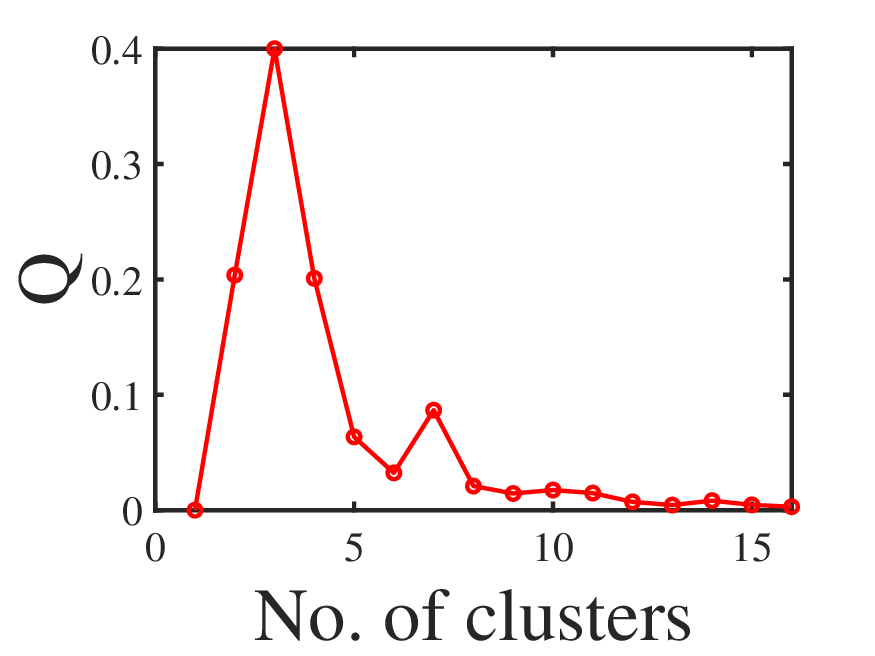}}
\subfigure[Karate-club-weighted]{\includegraphics[width=0.323\textwidth]{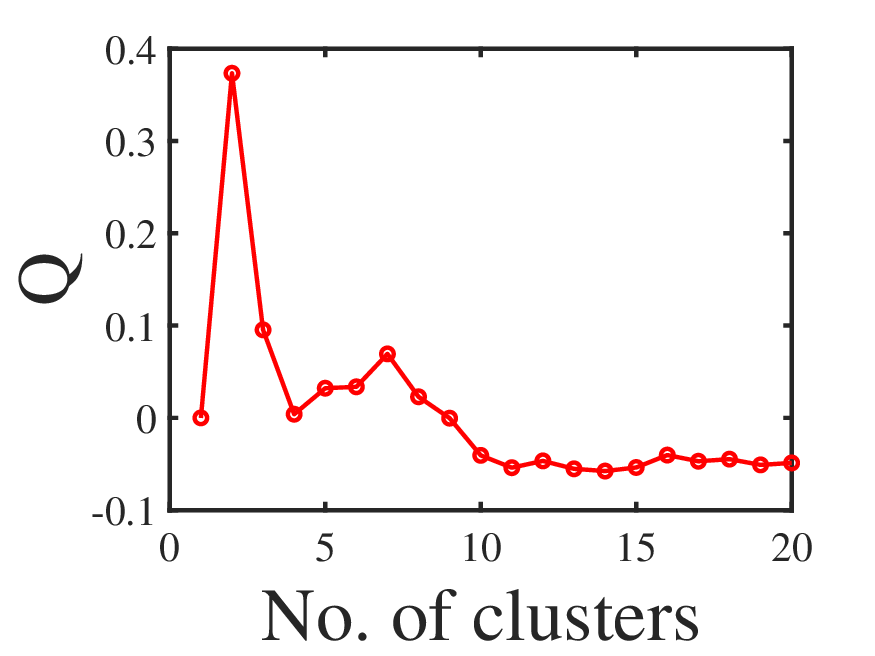}}
\subfigure[Slovene Parliamentary Party]{\includegraphics[width=0.323\textwidth]{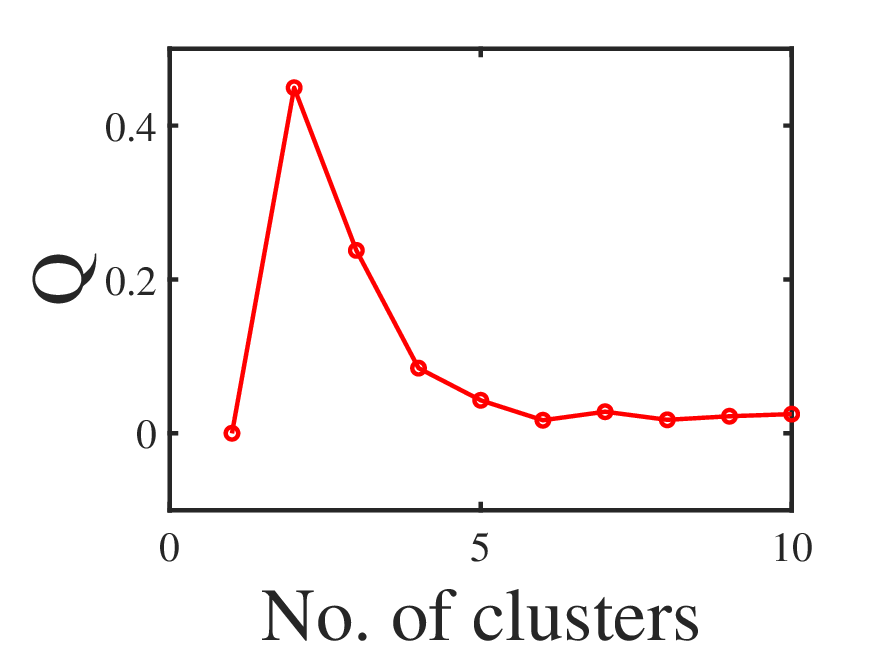}}
\subfigure[Train bombing]{\includegraphics[width=0.323\textwidth]{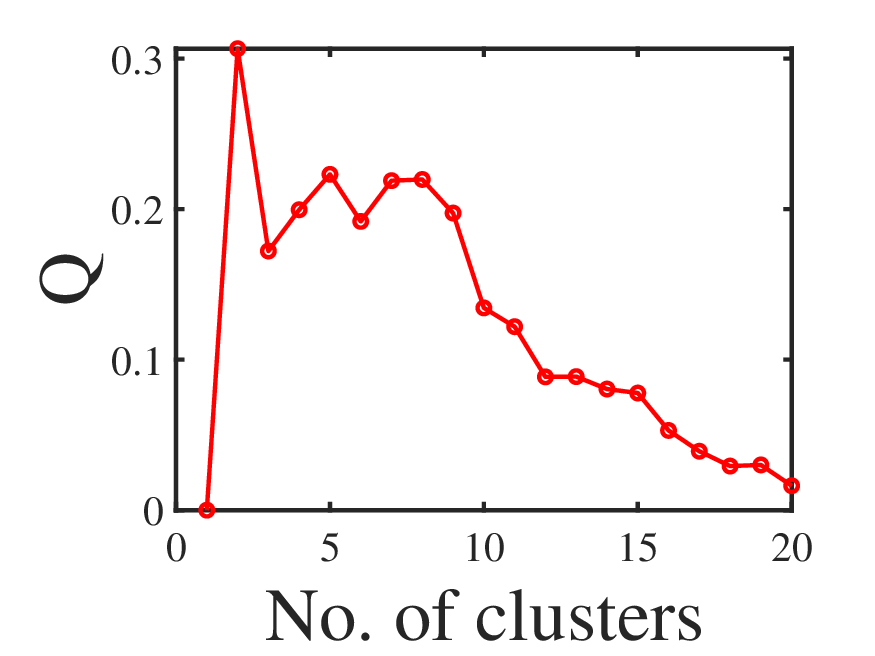}}
\subfigure[Les Mis\'erables]{\includegraphics[width=0.323\textwidth]{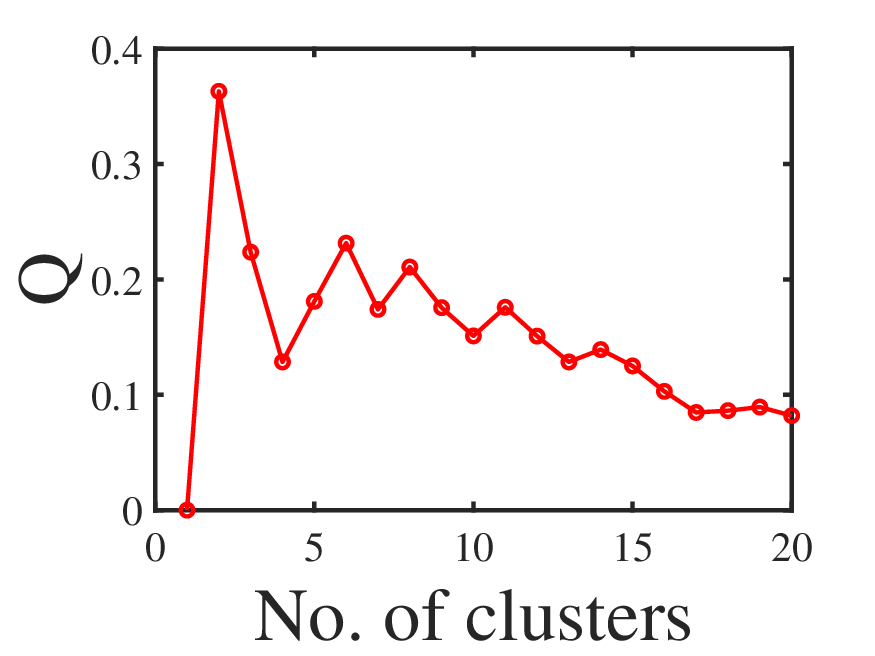}}
\subfigure[US Top-500 Airport Network]{\includegraphics[width=0.323\textwidth]{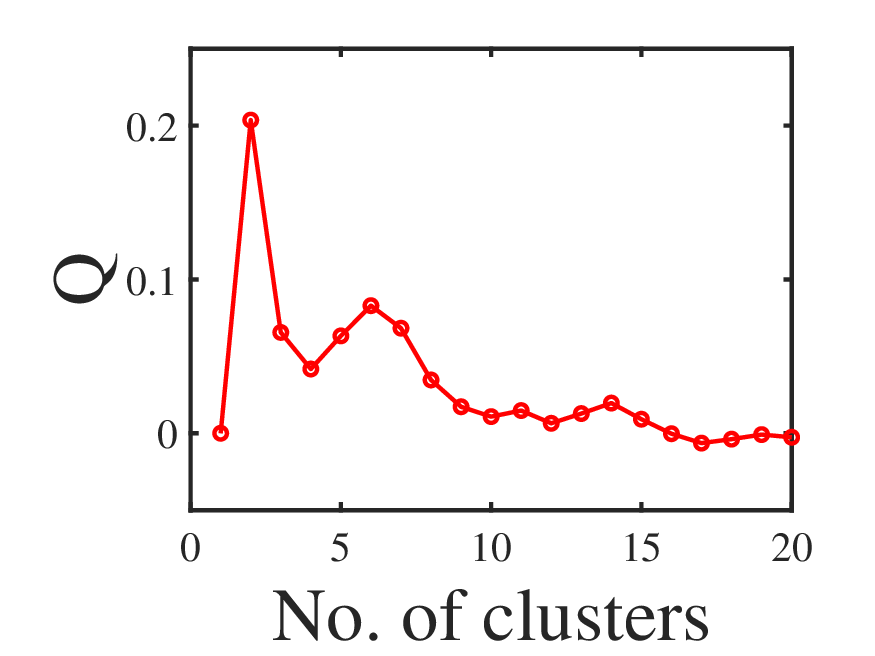}}
\subfigure[Political blogs]{\includegraphics[width=0.323\textwidth]{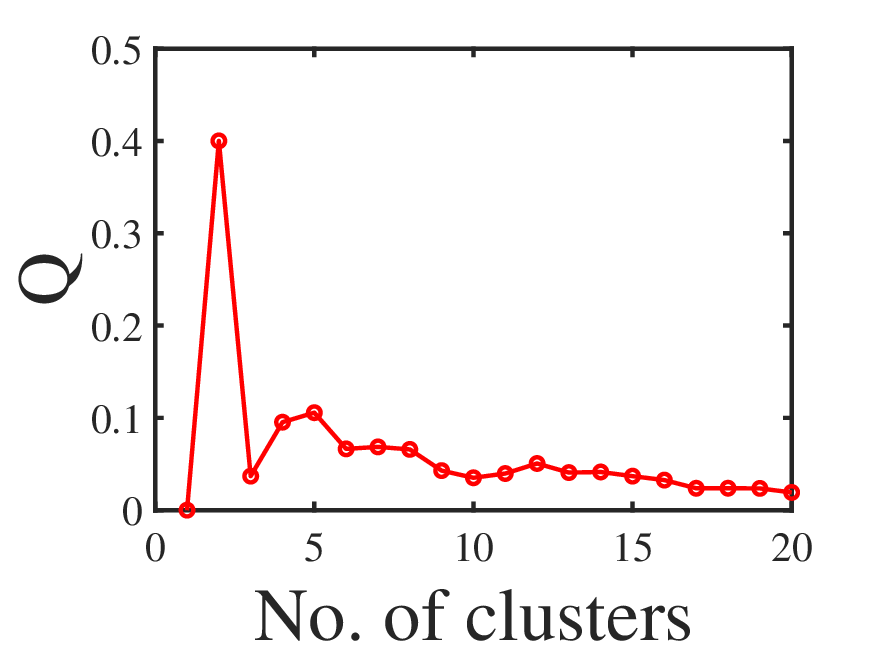}}
\subfigure[US airports]{\includegraphics[width=0.323\textwidth]{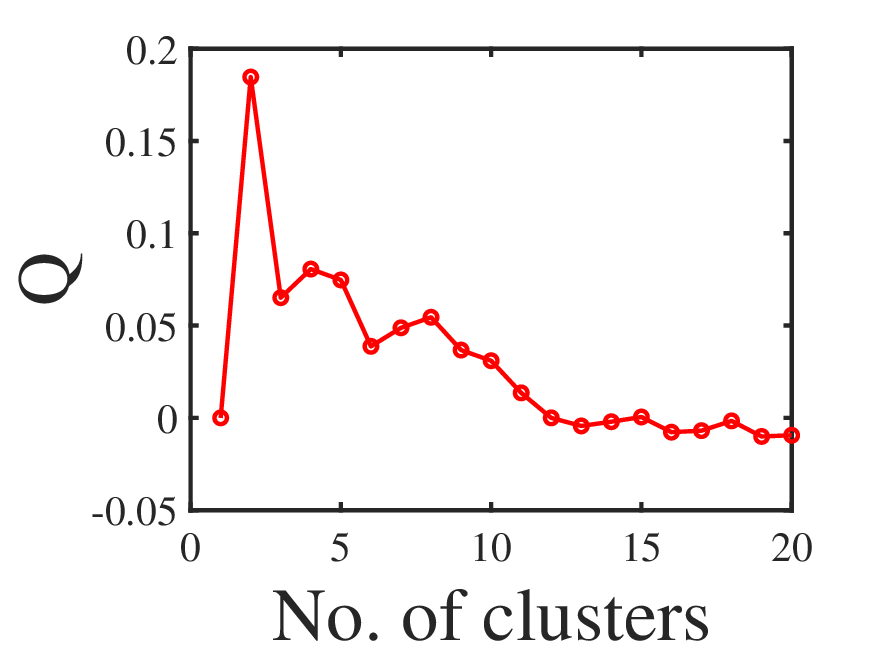}}
\subfigure[Cond-mat-1999]{\includegraphics[width=0.323\textwidth]{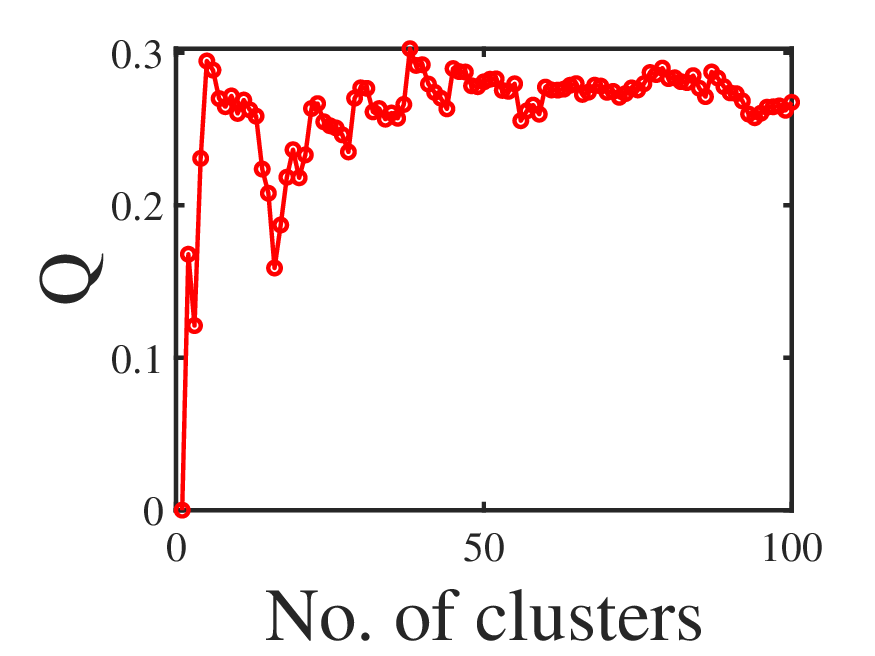}}
\caption{Fuzzy weighted modularity $Q$ computed by Equation (\ref{Modularity}) against the number of clusters by DFSP for real-world networks in Table \ref{realdata}.}
\label{Qdfsp} 
\end{figure}

\begin{table}[h!]
\footnotesize
	\centering
	\caption{Fuzzy weighted modularity $Q$ computed by Equation (\ref{Modularity}) of DFSP and its competitors for real-world networks used in this paper.}
	\label{Qalgorithms}
	\begin{tabular}{cccccccccccc}
\toprule
\textbf{Dataset}&\textbf{DFSP}&\textbf{GeoNMF}&\textbf{SVM-cD}&\textbf{OCCAM}\\
\midrule
Gahuku-Gama subtribes&\textbf{0.4000}&0.1743&0.1083&0.2574\\
Karate-club-weighted&0.3734&\textbf{0.3885}&0.3495&0.3685\\
Slovene Parliamentary Party&\textbf{0.4492}&0.3641&0.3678&0.4488\\
Train bombing&\textbf{0.3066}&0.2578&0.2604&0.2867\\
Les Mis\'erables&\textbf{0.3630}&0.3608&0.3014&0.3407\\
US Top-500 Airport Network&\textbf{0.2036}&0.1729&0.1345&0.1749\\
Political blogs&\textbf{0.4001}&\textbf{0.4001}&0.3789&0.3953\\
US airports&\textbf{0.1846}&0.1446&0.0453&0.1603\\
Cond-mat-1999&\textbf{0.3026}&0.2295&0.0192&0.2677\\
\botrule
\end{tabular}
\end{table}

To have a better understanding of the community structure for real-world networks, we define the following indices. Call node $i$ a highly mixed node if $\mathrm{max}_{k\in[K]}\hat{\Pi}(i,k)\leq0.7$ and a highly pure node if $\mathrm{max}_{k\in[K]}\hat{\Pi}(i,k)\geq0.9$. Let $\eta_{\mathrm{mixed}}$ be the proportion of highly mixed nodes and $\eta_{\mathrm{pure}}$ be the proportion of highly pure nodes in a network, where $\eta_{\mathrm{mixed}}$ and $\eta_{\mathrm{pure}}$ capture mixedness and purity of a network, respectively. The two indices for real data are displayed in Table \ref{realdataPureMixed} and we have the following conclusions:
\begin{itemize}
   \item For Gahuku-Gama subtribe, it has $16\times0.0625=1$ highly mixed node and $16\times0.8750=14$ highly pure nodes.
   \item For Karate-club-weighted, it has $34\times0.0588\approx2$ highly mixed nodes and $34\times0.7941\approx27$ highly pure nodes.
   \item For Slovene Parliamentary Party, it has 9 highly pure nodes and 0 highly mixed nodes.
   \item For Train bombing, it has $64\times0.0938\approx6$ highly mixed nodes and $64\times0.7969\approx51$ highly pure nodes.
   \item For Les Mis\'erables, it has $77\times0.0130\approx1$ highly mixed node and $77\times0.9351\approx72$ highly pure nodes.
   \item For US Top-500 Airport Network, it has $500\times0.1400=70$ highly mixed nodes and $500\times0.7820=391$ highly pure nodes.
   \item For Political blogs, it has $1222\times0.0393\approx48$ highly mixed nodes and $1222\times0.8781\approx1073$ highly pure nodes.
   \item For US airports, it has $1572\times0.0865\approx136$ highly mixed nodes and $1572\times0.8575\approx1348$ highly pure nodes.
   \item For Con-mat-1999, it has $13861\times0.6305\approx8739$ highly mixed nodes and $13861\times0.1558\approx2160$ highly pure nodes.
 \end{itemize}
\begin{table}[h!]
\footnotesize
	\centering
	\caption{$\eta_{\mathrm{mixed}}$ and $\eta_{\mathrm{pure}}$ obtained from DFSP for real-world networks considered in this paper.}
	\label{realdataPureMixed}
	\begin{tabular}{cccccccccccc}
\toprule
\textbf{Dataset}&\boldmath{$\eta_{\mathrm{mixed}}$}&\boldmath{$\eta_{\mathrm{pure}}$}\\
\midrule
Gahuku-Gama subtribes&0.0625&0.8750\\
Karate-club-weighted&0.0588&0.7941\\
Slovene Parliamentary Party&0&0.9\\
Train bombing&0.0938&0.7969\\
Les Mis\'erables&0.0130&0.9351\\
US Top-500 Airport Network&0.1400&0.7820\\
Political blogs&0.0393&0.8781\\
US airports&0.0865&0.8575\\
Cond-mat-1999&0.6305&0.1558\\
\botrule
\end{tabular}
\end{table}

For visibility, Figure \ref{NetReal} depicts communities returned by DFSP for the first eight small-size networks in Table \ref{realdata}, where we only highlight highly mixed nodes because most nodes are highly pure by Table \ref{realdataPureMixed}.
\begin{figure}
\centering
\normalsize
\subfigure[Gahuku-Gama subtribes]{\includegraphics[width=0.4\textwidth]{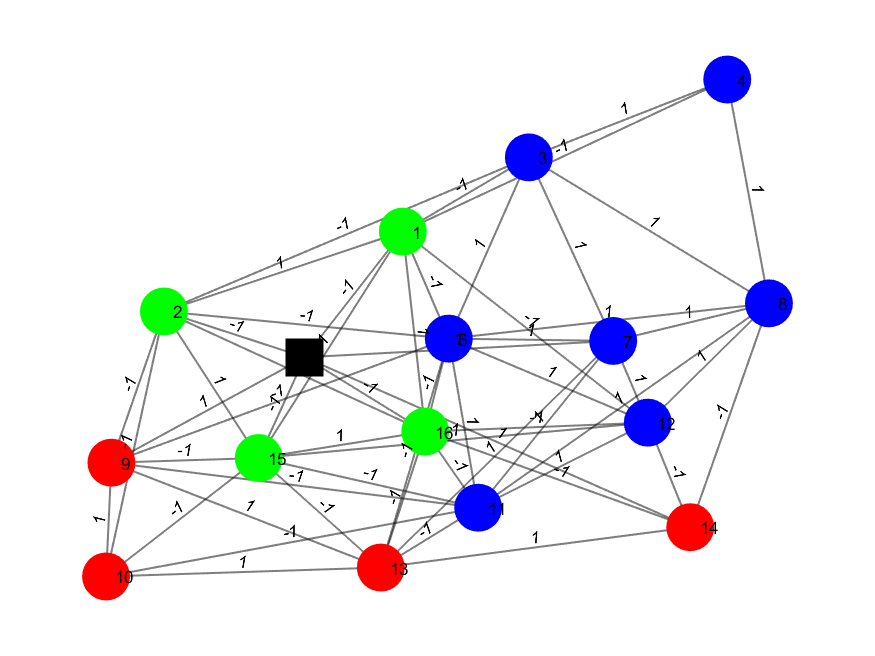}}
\subfigure[Karate-club-weighted]{\includegraphics[width=0.4\textwidth]{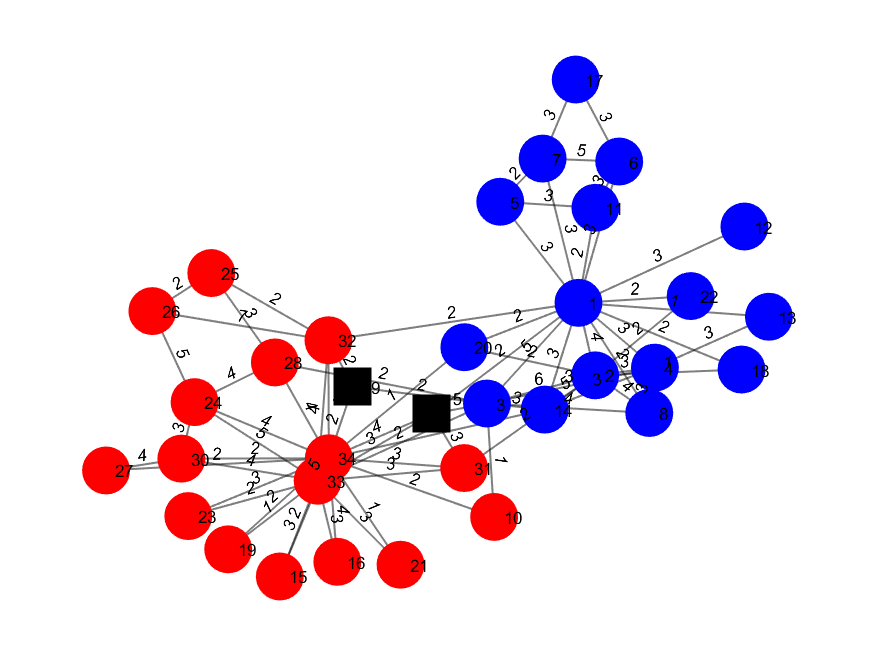}}
\subfigure[Slovene Parliamentary Party]{\includegraphics[width=0.4\textwidth]{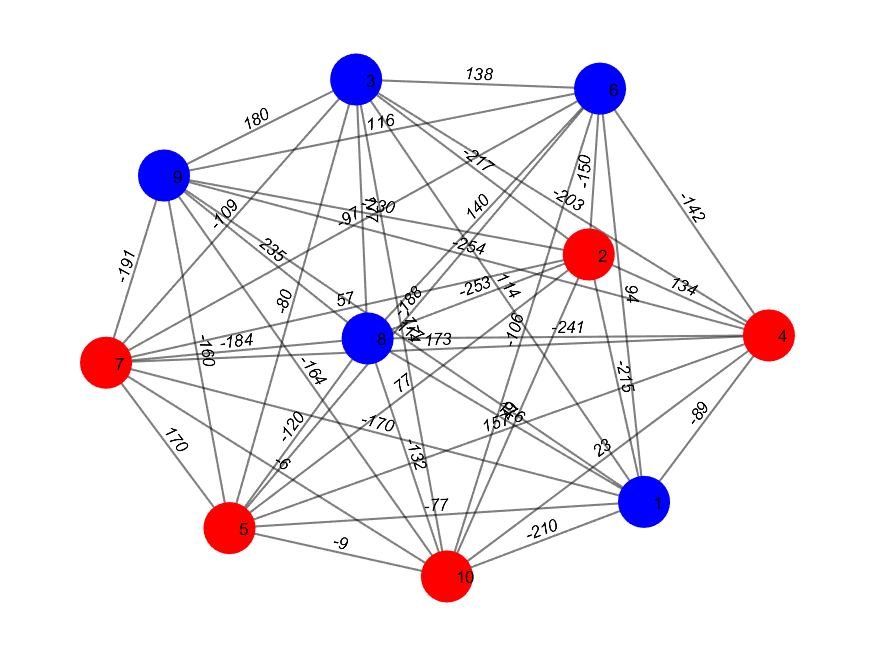}}
\subfigure[Train bombing]{\includegraphics[width=0.4\textwidth]{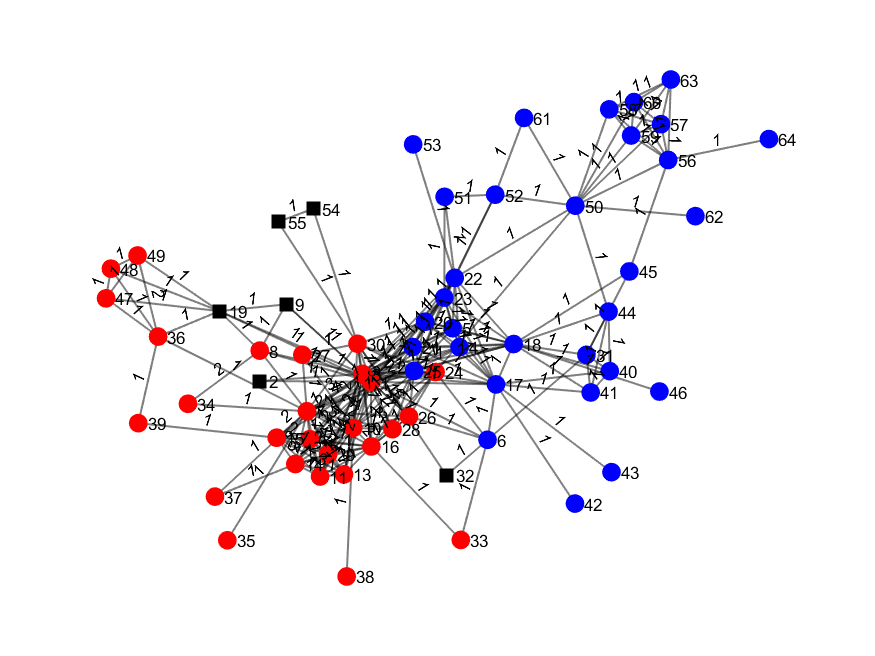}}
\subfigure[Les Mis\'erables]{\includegraphics[width=0.4\textwidth]{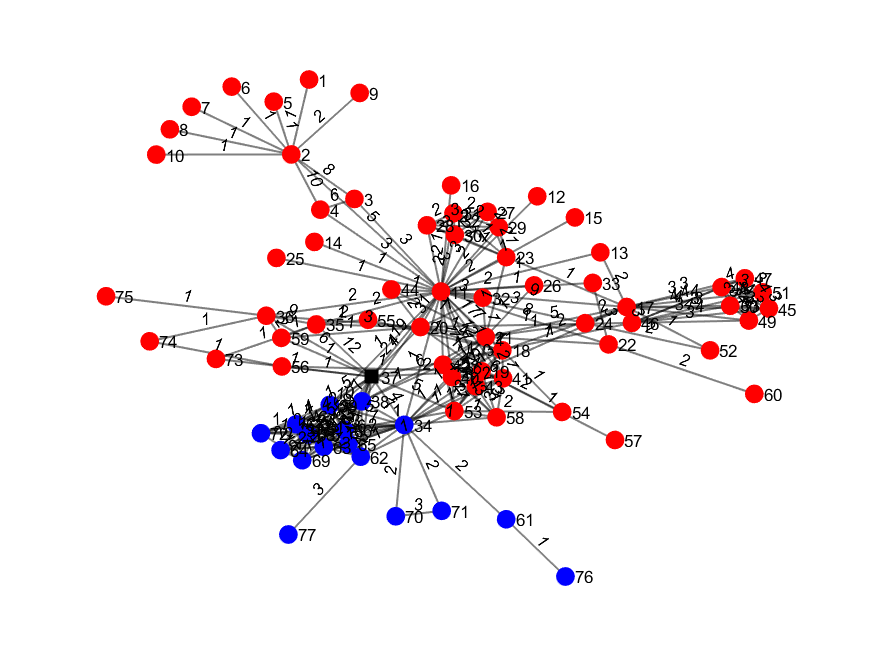}}
\subfigure[US Top-500 Airport Network]{\includegraphics[width=0.4\textwidth]{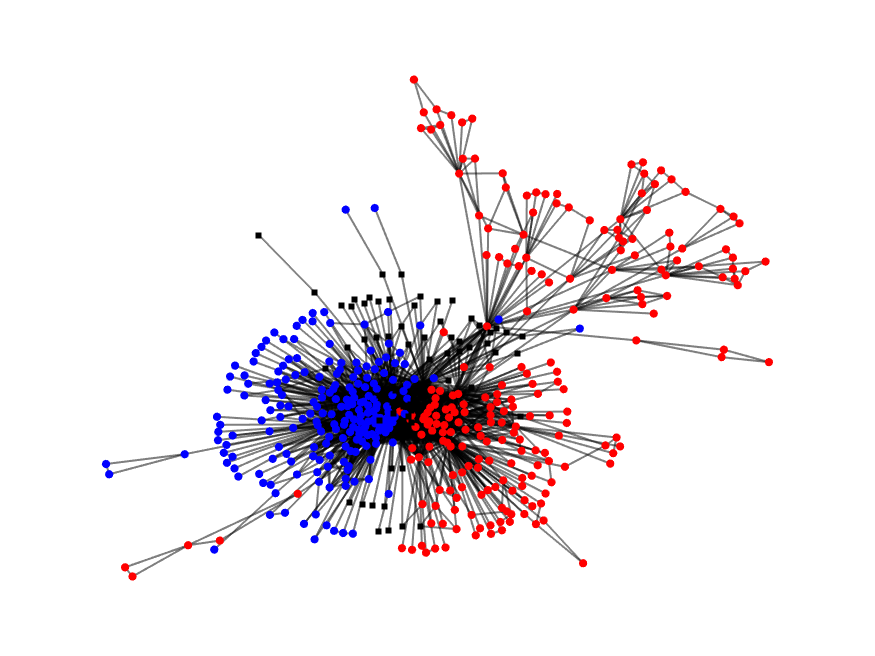}}
\subfigure[Political blogs]{\includegraphics[width=0.4\textwidth]{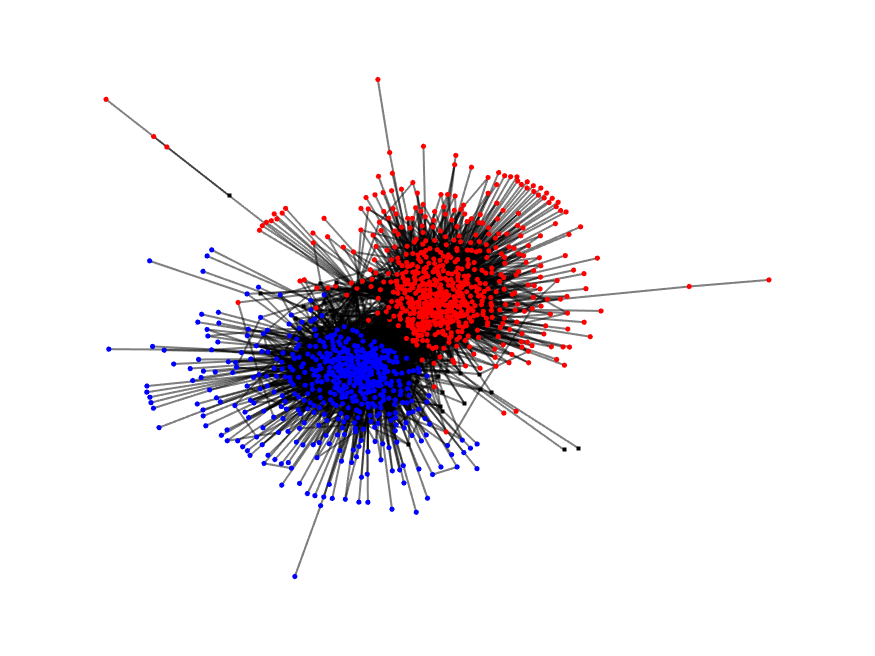}}
\subfigure[US airports]{\includegraphics[width=0.4\textwidth]{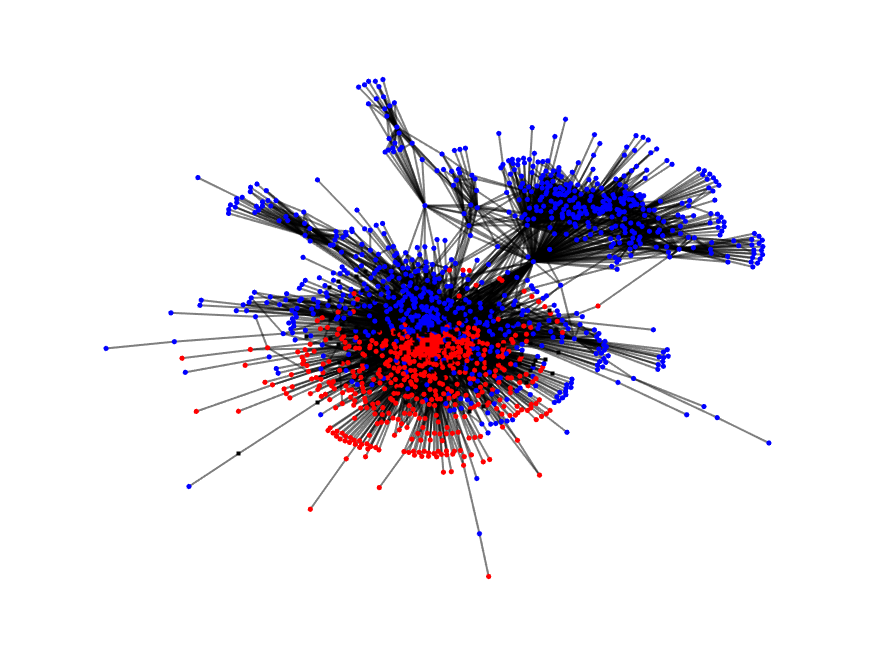}}
\caption{Communities detected by DFSP. Colors indicate communities and the black square indicates highly mixed nodes, where communities are obtained by $\hat{c}$. For visualization, we do not show edge weights and node labels for US Top-500 Airport Network, Political blogs, and US airports.}
\label{NetReal} 
\end{figure}
\section{Conclusion}\label{Conclusion}
In this paper, we have proposed a general, flexible, and identifiable mixed membership distribution-free (MMDF) model to capture community structures of overlapping weighted networks. An efficient spectral algorithm, DFSP, was used to conduct mixed membership community detection and shown to be consistent under mild conditions in the MMDF framework. We have also proposed the fuzzy weighted modularity for overlapping weighted networks. And by maximizing the fuzzy weighted modularity, we can get an efficient estimation of the number of communities for overlapping weighted networks. The advantages of MMDF and fuzzy weighted modularity are validated on both computer-generated and real-world weighted networks. Experimental results demonstrated that DFSP outperforms its competitors in community detection and KDFSP outperforms its competitors in inferring the number of communities.

MMDF is a generative model and fuzzy weighted modularity is a general modularity for overlapping weighted networks. We expect that our model MMDF and fuzzy weighted modularity proposed in this paper will have wide applications in learning and understanding the latent structure of overlapping weighted networks, just as the mixed membership stochastic blockmodels and the Newman-Girvan modularity which have been widely studied in recent years. Future works will be studied from four aspects: the first is studying the variation of MMDF for heterogeneous networks; the second is presenting a rigorous method to determine the number of communities for networks generated from MMDF; the third is developing algorithms to estimate community membership based on our fuzzy weighted modularity; the fourth is developing a general framework of MMDF to deal with large-scale weighted networks. One possible idea is, we can accelerate the eigendecomposition step of the DFSP method by the two randomized sketching techniques (the random projection-based technique and the random sampling-based technique) considered in \cite{zhang2022randomized} to process large-scale networks. Theoretical guarantees for consistent estimation of the accelerated algorithm can also be analyzed under the framework of our MMDF model.
\bmhead{Acknowledgements} Wang’s work was supported by the Fundamental Research Funds for the Central Universities, Nankai Univerity, 63231186 and the National Natural Science Foundation of China (Grant 12001295, 12271272).
\bmhead{Author Contributions} Huan Qing: Conceptualization, Methodology, Investigation, Software, Formal analysis, Data curation, Writing-original draft, Writing-reviewing \& editing. Jingli Wang: Writing-reviewing \& editing, Funding acquisition.
\section*{Declarations}
\textbf{Conflict of interest} The authors declare no conflict of interest.
\begin{appendices}
\section{Vertex hunting algorithm}\label{VHAandExactly}
Algorithm \ref{alg:SP} is the SP algorithm.
	\begin{algorithm}
		\caption{\textbf{Successive Projection (SP)} \cite{gillis2015semidefinite}}
		\label{alg:SP}
		\begin{algorithmic}[1]
			\Require Near-separable matrix $Y_{sp}=S_{sp}M_{sp}+Z_{sp}\in\mathbb{R}^{m\times n}_{+}$ , where $S_{sp}, M_{sp}$ should satisfy Assumption 1 \cite{gillis2015semidefinite}, the number $r$ of columns to be extracted.
			\Ensure Set of indices $\mathcal{K}$ such that $Y_{sp}(\mathcal{K},:)\approx S$ (up to permutation)
			\State Let $R=Y_{sp}, \mathcal{K}=\{\}, k=1$.
			\State \textbf{While} $R\neq 0$ and $k\leq r$ \textbf{do}
			\State ~~~~~~~$k_{*}=\mathrm{argmax}_{k}\|R(k,:)\|_{F}$.
			\State ~~~~~~$u_{k}=R(k_{*},:)$.
			\State ~~~~~~$R\leftarrow (I-\frac{u_{k}u'_{k}}{\|u_{k}\|^{2}_{F}})R$.
			\State ~~~~~~$\mathcal{K}=\mathcal{K}\cup \{k_{*}\}$.
			\State ~~~~~~k=k+1.
			\State \textbf{end while}
		\end{algorithmic}
	\end{algorithm}
\section{Proofs under MMDF}
\subsection{Proof of Proposition \ref{idMMDF}}
\begin{proof}
This proposition holds immediately by the first statement of Theorem 2.1 \cite{mao2020estimating} since we let $P$ be a full rank matrix and Theorem 2.1 \cite{mao2020estimating} is a distribution-free result such that it always holds without constraining the distribution of $A$.
\end{proof}
\subsection{Proof of Lemma \ref{IS}}
\begin{proof}
Since $\Omega=\Pi \rho P\Pi'=U\Lambda U'$ and $U'U=I_{K}$, we have $U=\Pi \rho P\Pi'U\Lambda^{-1}$, i.e., $B=\rho P\Pi' U\Lambda^{-1}$. So $B$ is unique. Since $U=\Pi B$, we have $U(\mathcal{I},:)=\Pi(\mathcal{I},:)B=B$ and the lemma follows.
\end{proof}
\subsection{Proof of Theorem \ref{Main}}
\begin{proof}
First, we prove the following lemma to provide an upper bound of row-wise eigenspace error $\|\hat{U}\hat{U}'-UU'\|_{2\rightarrow\infty}$.
\begin{lem}\label{rowwiseerror}
(Row-wise eigenspace error) Under $MMDF_{n}(K,P,\Pi,\rho,\mathcal{F})$, when Assumption \ref{assumesparsity} holds, suppose $\sigma_{K}(\Omega)\geq C\sqrt{\gamma\rho n\mathrm{log}(n)}$ for some $C>0$, with probability at least $1-o(n^{-3})$, we have
\begin{align*}
\|\hat{U}\hat{U}'-UU'\|_{2\rightarrow\infty}=O(\frac{\sqrt{\gamma n\mathrm{log}(n)}}{\sigma_{K}(P)\rho^{0.5} \lambda^{1.5}_{K}(\Pi'\Pi)}).
\end{align*}
\end{lem}
\begin{proof}
First, we use Theorem 1.4 (the Matrix Bernstein) of \cite{tropp2012user} to build an upper bound of $\|A-\Omega\|_{\infty}$. This theorem is given below
\begin{thm}\label{Bern}
Consider a finite sequence $\{X_{k}\}$ of independent, random, self-adjoint matrices with dimension $d$. Assume that each random matrix satisfies
\begin{align*}
\mathbb{E}[X_{k}]=0, \mathrm{and~}\|X_{k}\|\leq R~\mathrm{almost~surely}.
\end{align*}
Then, for all $t\geq 0$,
\begin{align*}
\mathbb{P}(\|\sum_{k}X_{k}\|\geq t)\leq d\cdot \mathrm{exp}(\frac{-t^{2}/2}{\sigma^{2}+Rt/3}),
\end{align*}
where $\sigma^{2}:=\|\sum_{k}\mathbb{E}(X^{2}_{k})\|$.
\end{thm}
Let $x=(x_{1},x_{2},\ldots, x_{n})'$ be any $n\times 1$ vector. For any $i,j\in[n]$, we have $\mathbb{E}[(A(i,j)-\Omega(i,j))x(j)]=0$ and $\|(A(i,j)-\Omega(i,j))x(j)\|\leq \tau\|x\|_{\infty}$. Set $R=\tau\|x\|_{\infty}$. Since $\|\sum_{j=1}^{n}\mathbb{E}[(A(i,j)-\Omega(i,j))^{2}x^{2}(j)]\|=\|\sum_{j=1}^{n}x^{2}(j)\mathbb{E}[(A(i,j)-\Omega(i,j))^{2}]\|=\|\sum_{j=1}^{n}x^{2}(j)\mathrm{Var}(A(i,j))\|\leq \gamma\rho\sum_{j=1}^{n}x^{2}(j)$, by Theorem \ref{Bern}, for any $t\geq 0$ and $i\in[n]$, we have
\begin{align*}
\mathbb{P}(|\sum_{j=1}^{n}(A(i,j)-\Omega(i,j))x(j)|>t)\leq 2\mathrm{exp}(-\frac{t^{2}/2}{\gamma\rho\sum_{j=1}^{n}x^{2}(j)+\frac{Rt}{3}}).
\end{align*}
Set $x(j)$ as $1$ or $-1$ such that $(A(i,j)-\Omega(i,j))y(j)=|A(i,j)-\Omega(i,j)|$, we have
\begin{align*}
\mathbb{P}(\|A-\Omega\|_{\infty}>t)\leq 2\mathrm{exp}(-\frac{t^{2}/2}{\gamma\rho n+\frac{Rt}{3}}).
\end{align*}
Set $t=\frac{\alpha+1+\sqrt{(\alpha+1)(\alpha+19)}}{3}\sqrt{\gamma\rho n\mathrm{log}(n)}$ for any $\alpha>0$. By assumption \ref{assumesparsity}, we have
\begin{align*}
\mathbb{P}(\|A-\Omega\|_{\infty}>t)\leq 2\mathrm{exp}(-\frac{t^{2}/2}{\gamma\rho n+\frac{Rt}{3}})\leq n^{-\alpha}.
\end{align*}
By Theorem 4.2 of \cite{cape2019the}, when $\sigma_{K}(\Omega)\geq 4\|A-\Omega\|_{\infty}$, we have
\begin{align*}
\|\hat{U}-U\mathcal{O}\|_{2\rightarrow\infty}\leq 14\frac{\|A-\Omega\|_{\infty}}{\sigma_{K}(\Omega)}\|U\|_{2\rightarrow\infty},
\end{align*}
where $\mathcal{O}$ is a $K\times K$ orthogonal matrix. With probability at least $1-o(n^{-\alpha})$, we have
\begin{align*}
\|\hat{U}-U\mathcal{O}\|_{2\rightarrow\infty}=O(\frac{\|U\|_{2\rightarrow\infty}\sqrt{\gamma\rho n\mathrm{log}(n)}}{\sigma_{K}(\Omega)}).
\end{align*}
Since $\hat{U}'\hat{U}=I_{K},U'U=I_{K}$, by basic algebra, we have $\|\hat{U}\hat{U}'-UU'\|_{2\rightarrow\infty}\leq2\|\hat{U}-U\mathcal{O}\|_{2\rightarrow\infty}$, which gives
\begin{align*}
\|\hat{U}\hat{U}'-UU'\|_{2\rightarrow\infty}=O(\frac{\|U\|_{2\rightarrow\infty}\sqrt{\gamma\rho n\mathrm{log}(n)}}{\sigma_{K}(\Omega)}).
\end{align*}
Since $\sigma_{K}(\Omega)\geq \sigma_{K}(P)\rho \lambda_{K}(\Pi'\Pi)$ by Lemma II.4 of \cite{mao2020estimating} and $\|U\|^{2}_{2\rightarrow\infty}\leq \frac{1}{\lambda_{K}(\Pi'\Pi)}$ by Lemma 3.1 of \cite{mao2020estimating}, where these two lemmas are distribution-free and always hold as long as Equations (\ref{DefinePI}), (\ref{definP}) and (\ref{DefinOmega}) hold, we have
\begin{align*}
\|\hat{U}\hat{U}'-UU'\|_{2\rightarrow\infty}=O(\frac{\sqrt{\gamma n\mathrm{log}(n)}}{\sigma_{K}(P)\rho^{0.5} \lambda^{1.5}_{K}(\Pi'\Pi)}).
\end{align*}
Set $\alpha=3$, and this claim follows.
\begin{rem}
Alternatively, Theorem 4.2. of \cite{chen2021spectral} can also be applied to obtain the upper bound of $\|\hat{U}\hat{U}'-UU'\|_{2\rightarrow\infty}$, and this bound is similar to the one in Lemma \ref{rowwiseerror}.
\end{rem}
\end{proof}
For convenience, set $\varpi=\|\hat{U}\hat{U}'-UU'\|_{2\rightarrow\infty}$. Since DFSP is the SPACL algorithm without the prune step of \cite{mao2020estimating}, the proof of DFSP's consistency is the same as SPACL except for the row-wise eigenspace error step where we need to consider $\gamma$ which is directly related with distribution $\mathcal{F}$. By Lemma \ref{rowwiseerror} and Equation (3) in Theorem 3.2 of \cite{mao2020estimating} where the proof is distribution-free, there exists a $K\times K$ permutation matrix $\mathcal{P}$ such that
\begin{align*}	\mathrm{max}_{i\in[n]}\|e'_{i}(\hat{\Pi}-\Pi\mathcal{P})\|_{1}=O(\varpi \kappa(\Pi'\Pi)\sqrt{\lambda_{1}(\Pi'\Pi)})=O(\frac{\kappa^{1.5}(\Pi'\Pi)\sqrt{\gamma n\mathrm{log}(n)}}{\sigma_{K}(P)\rho^{0.5} \lambda_{K}(\Pi'\Pi)}).
\end{align*}
\end{proof}
\subsection{Proof of Corollary \ref{AddConditions}}
\begin{proof}
When $\lambda_{K}(\Pi'\Pi)=O(\frac{n}{K})$ and $K=O(1)$, we have $\kappa(\Pi'\Pi)=O(1)$ and $\lambda_{K}(\Pi'\Pi)=O(n/K)=O(n)$. Then the corollary follows immediately by Theorem \ref{Main}.
\end{proof}
\section{Extra simulation results}
In this part, we consider two extra simulations: imbalanced networks and running time. For imbalanced networks, we study the stability of DFSP and its competitors when there are small-size communities. For running time, we compare the running time for each method by increasing the network size $n$. For simplicity, we only consider the case when $\mathcal{F}$ is Normal distribution here. When $A(i,j)\sim\mathrm{Normal}(\Omega(i,j),\sigma^{2}_{A})$, let all nodes be pure, $K=2, \rho=1$, and $\sigma^{2}_{A}=1$. Set $P$ as
\[P=\begin{bmatrix}
    1&-0.2\\
    -0.2&0.9\\
\end{bmatrix}.\]
Let the first community has $\delta n$ nodes. So, the second community has $(1-\delta)n$ nodes. Based on the above settings, we consider the following two simulations.

\textbf{Changing $\delta$}: Let $n=200$ or $n=1000$. Let $\delta$ range in $\{0.025, 0.05, 0.075, \ldots, 0.5\}$. For this case, the two evaluation metrics Hamming error and Relative are not suitable for imbalanced networks. To prioritize the ability of DFSP and its competitors to detect the minority communities, we consider the following two metrics who are the smaller the better.
\begin{align*}
&\mathrm{Clustering}~l_{1}~\mathrm{error}=\mathrm{min}_{\mathcal{P}\in\{ K\times K\mathrm{~permutation~matrix}\}}\mathrm{max}_{k\in[K]}\frac{\|\hat{\Pi}(:,k)-(\Pi\mathcal{P})(:,k)\|_{1}}{\|(\Pi\mathcal{P})(:,k)\|_{1}},\\
&\mathrm{Clustering}~l_{2}~\mathrm{error}=\mathrm{min}_{\mathcal{P}\in\{K\times K\mathrm{~permutation~matrix}\}}\mathrm{max}_{k\in[K]}\frac{\|\hat{\Pi}(:,k)-(\Pi\mathcal{P})(:,k)\|_{F}}{\|(\Pi\mathcal{P})(:,k)\|_{F}}.
\end{align*}

Unlike Hamming error which measures the $l_{1}$ difference between $\Pi$  and $\hat{\Pi}$ up to a permutation of community labels, Clustering $l_{1}$ error measures the maximum $l_{1}$ difference between the size of the true $k$-th community and the size of the estimated $k$-th community up to a permutation of community labels among all $K$ communities. Therefore, Clustering $l_{1}$ error can evaluate the ability of a community detection method to detect the minority communities. Similar arguments hold for the Clustering $l_{2}$ error.

Panels (a)-(f) of Figure \ref{EXTRASim} display numerical results for changing $\delta$. For the case when $n=200$, we find that DFSP and its competitors perform similarly and all of them can successfully detect the minority community when $\delta\in[0.125,0.5]$, i.e., the proportion of community sizes between the largest community and the smallest community locates in $[1,7]$. KDFSP successfully estimates the number of communities $K$ when $\delta\in[0.15,0.5]$ while NB and BHac fail to infer $K$ for all cases. For the case when $n=1000$, DFSP and its competitors successfully detect all communities when $\delta\in[0.1, 0.5]$, i.e., the proportion of community sizes between the largest community and the smallest community locates in $[1,9]$. KDFSP correctly determines $K$ when $\delta\in[0.05,0.5]$ while its competitors fail to find $K$.

\textbf{Changing $n$}: Let $\delta=0.075$ or $\delta=0.1$, i.e., let the proportion of community sizes between the largest community and the smallest community be $\frac{37}{3}$ or 9. Let $n$ range in $\{2000,4000,6000,\ldots,12000\}$. For simplicity, we only report the averaged Clustering $l_{1}$ error, averaged Clustering $l_{2}$ error, and averaged running time over 100 repetitions for DFSP and its competitors. Figure \ref{EXTRASimN} displays the numerical results. We see that DFSP is better than GeoNMF, SVM-cD, and OCCAM in both estimation accuracy and running time. In particular, DFSP runs much faster than OCCAM. Meanwhile, DFSP performs satisfactorily
for its small clustering errors for the two cases $\delta=0.075$ and $\delta=0.1$. By comparing panel (a) and panel (e) (panel (b) and panel (f)), we see that all methods perform poorer for a more imbalanced network and this result is consistent with that of changing $\delta$.  By comparing panel (c) and panel (g) (panel (d) and panel (h)), we see that each method takes more time to detect a more imbalanced network.

\begin{figure}
\centering
\resizebox{\columnwidth}{!}{
\subfigure[$n=200$]{\includegraphics[width=0.35\textwidth]{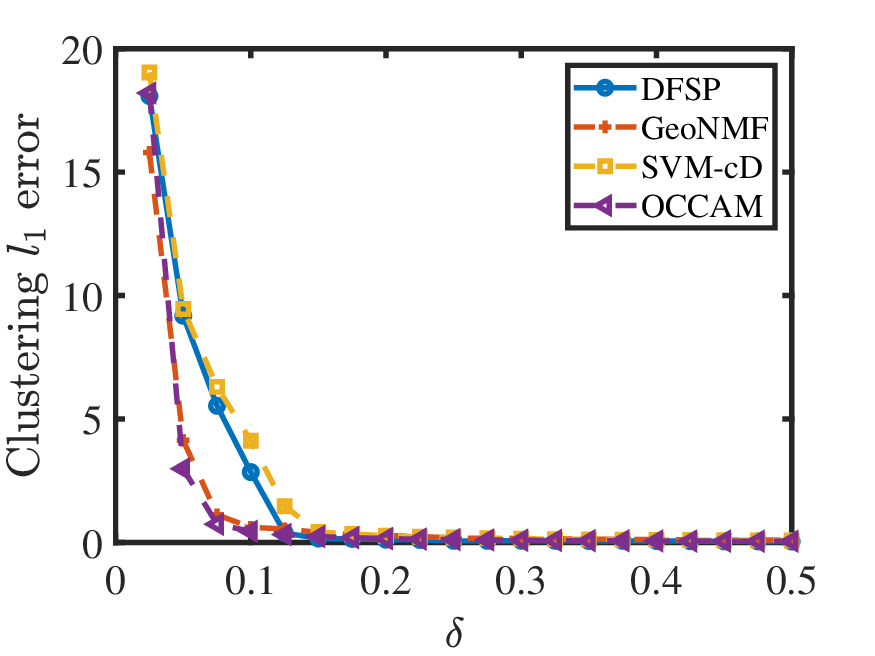}}
\subfigure[$n=200$]{\includegraphics[width=0.35\textwidth]{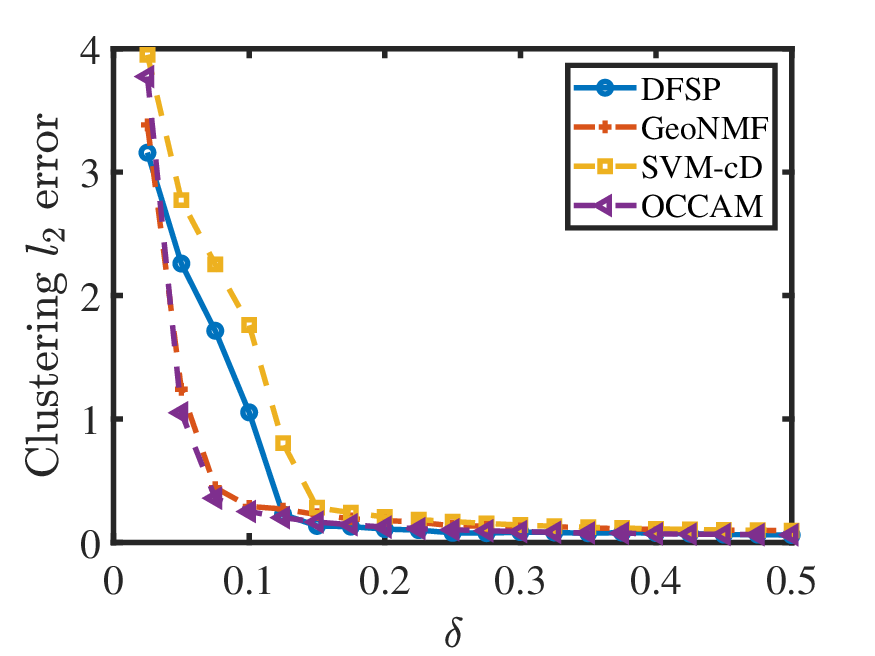}}
\subfigure[$n=200$]{\includegraphics[width=0.35\textwidth]{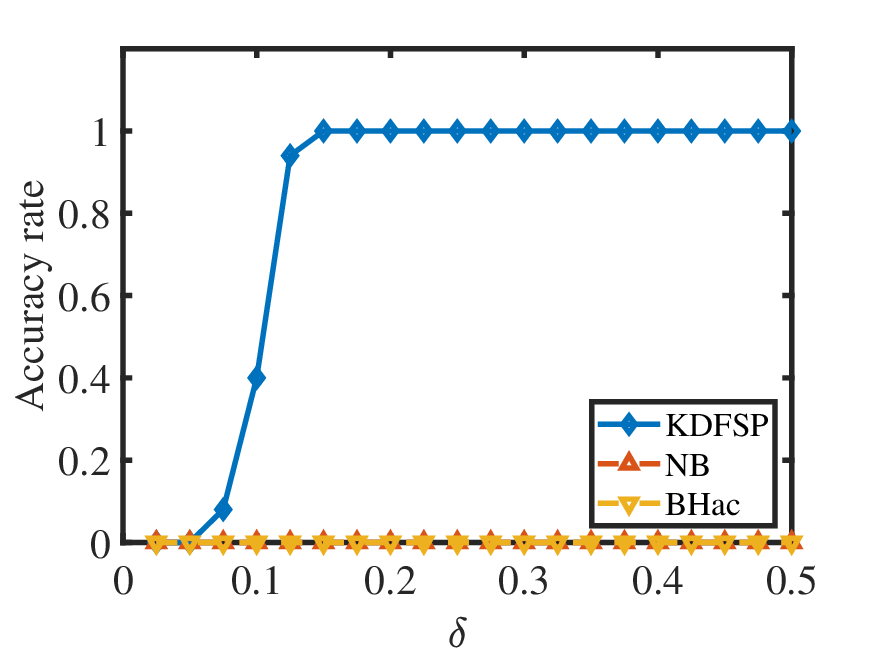}}
}
\resizebox{\columnwidth}{!}{
\subfigure[$n=1000$]{\includegraphics[width=0.35\textwidth]{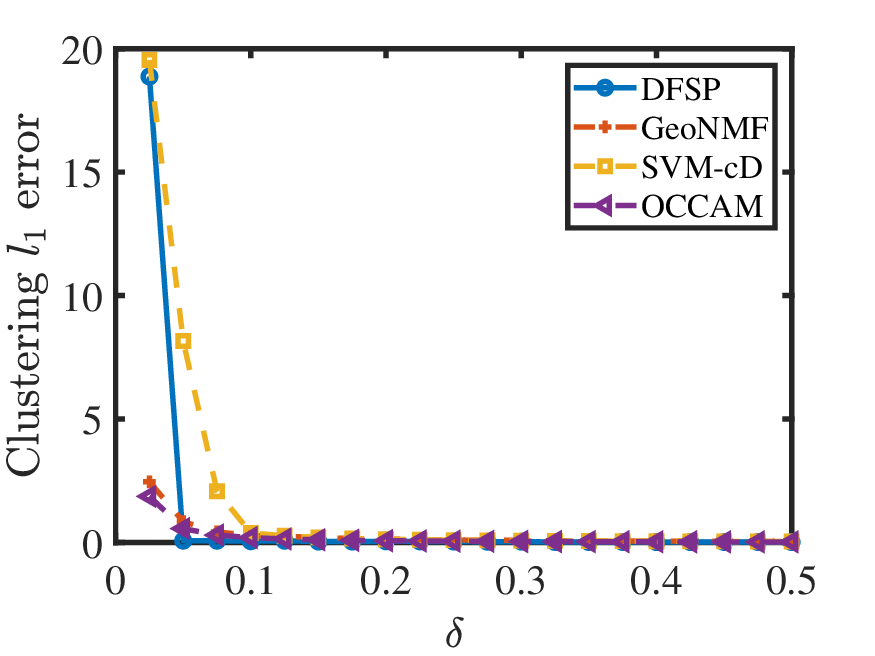}}
\subfigure[$n=1000$]{\includegraphics[width=0.35\textwidth]{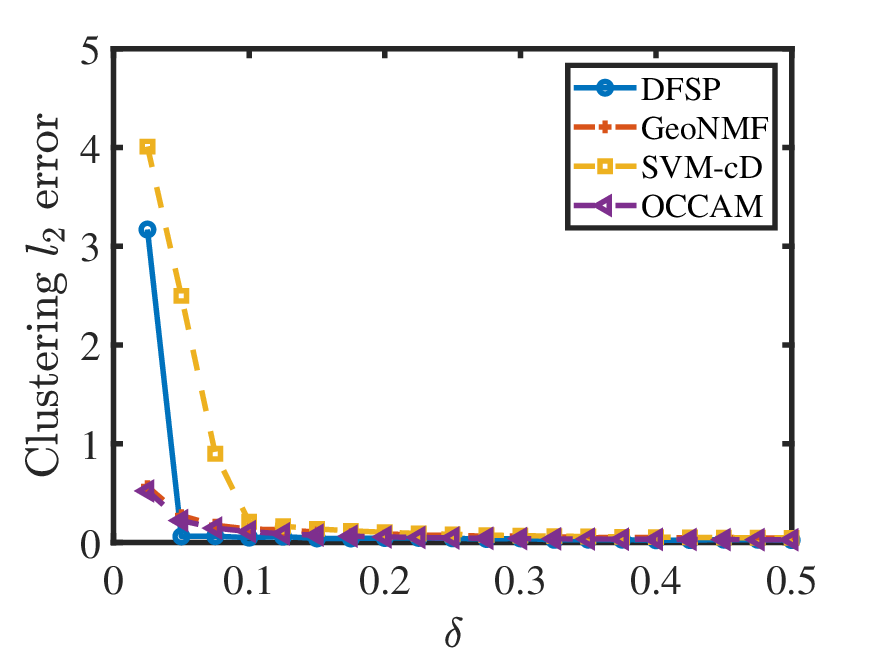}}
\subfigure[$n=1000$]{\includegraphics[width=0.35\textwidth]{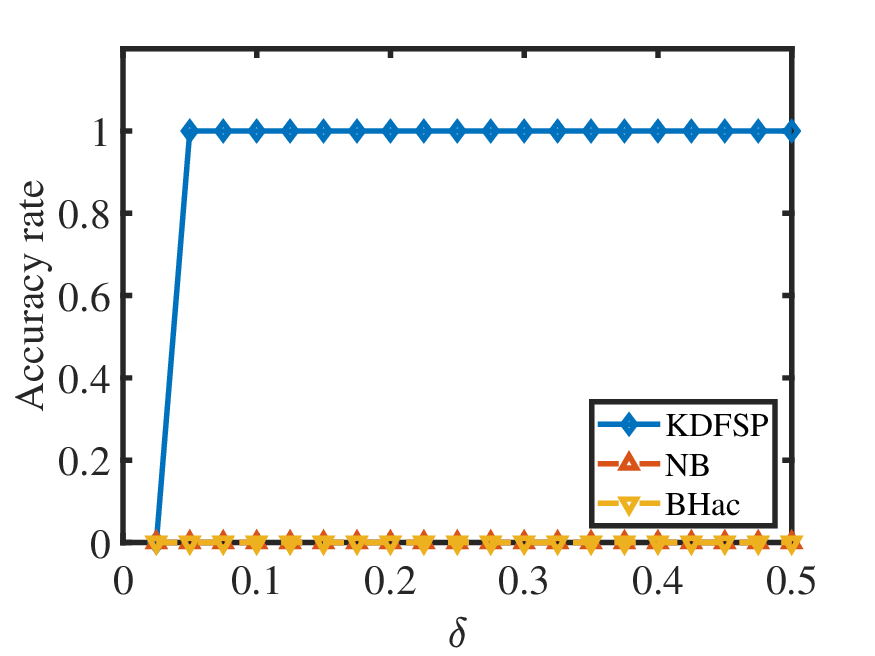}}
}
\caption{Numerical results of changing $\delta$.}
\label{EXTRASim} 
\end{figure}

\begin{figure}
\centering
\resizebox{\columnwidth}{!}{
\subfigure[$\delta=0.075$]{\includegraphics[width=0.35\textwidth]{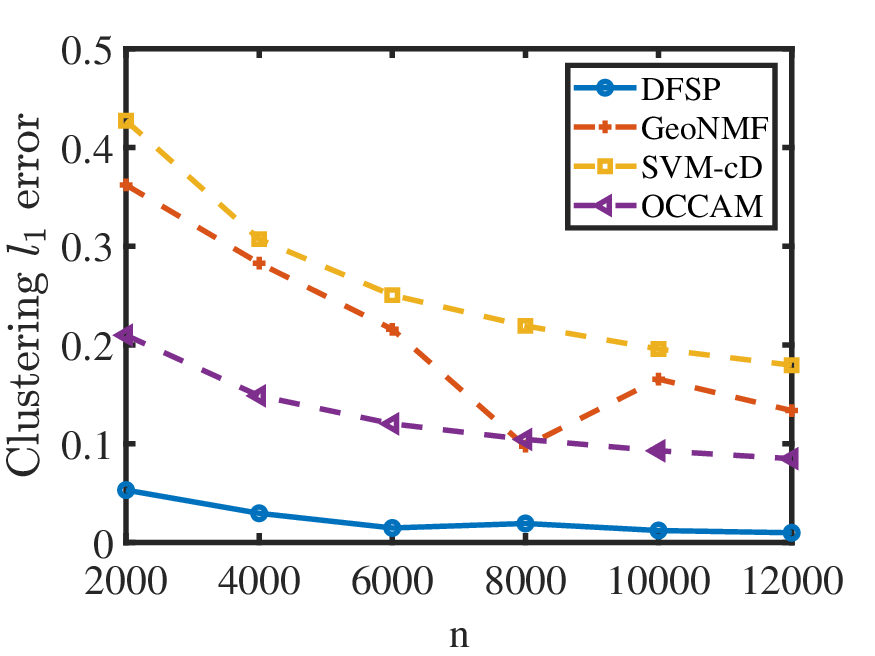}}
\subfigure[$\delta=0.075$]{\includegraphics[width=0.35\textwidth]{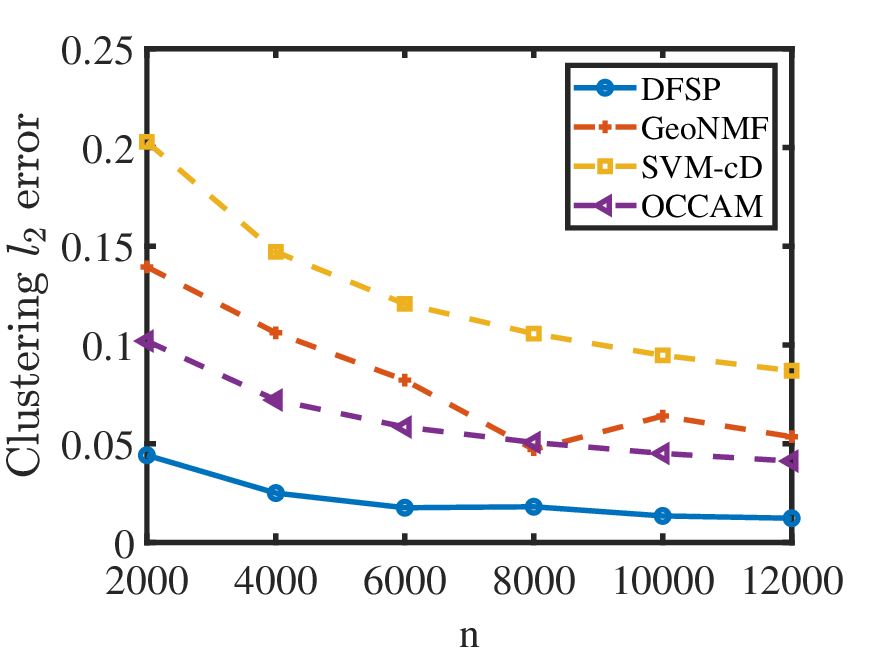}}
\subfigure[$\delta=0.075$]{\includegraphics[width=0.35\textwidth]{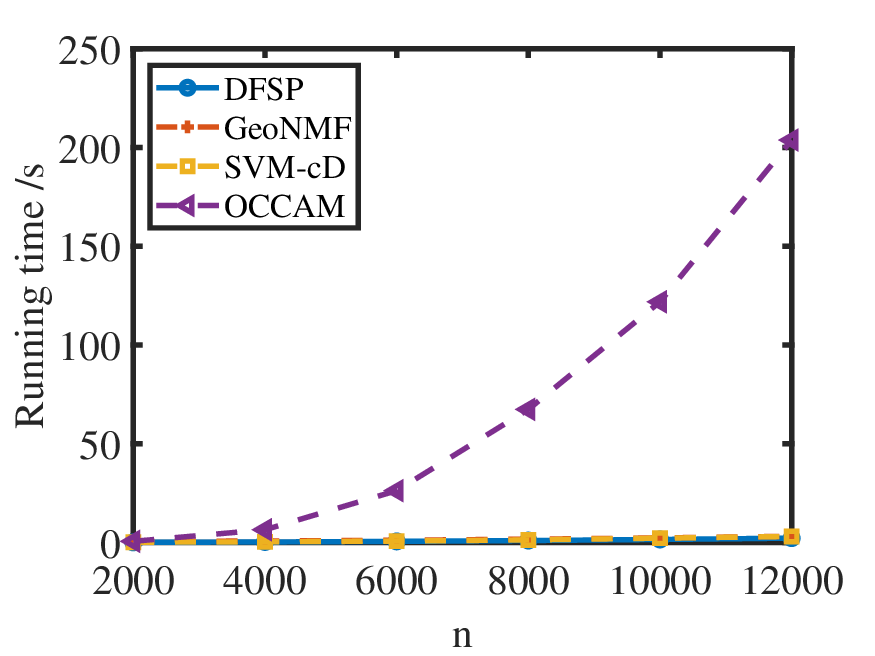}}
\subfigure[$\delta=0.075$]{\includegraphics[width=0.35\textwidth]{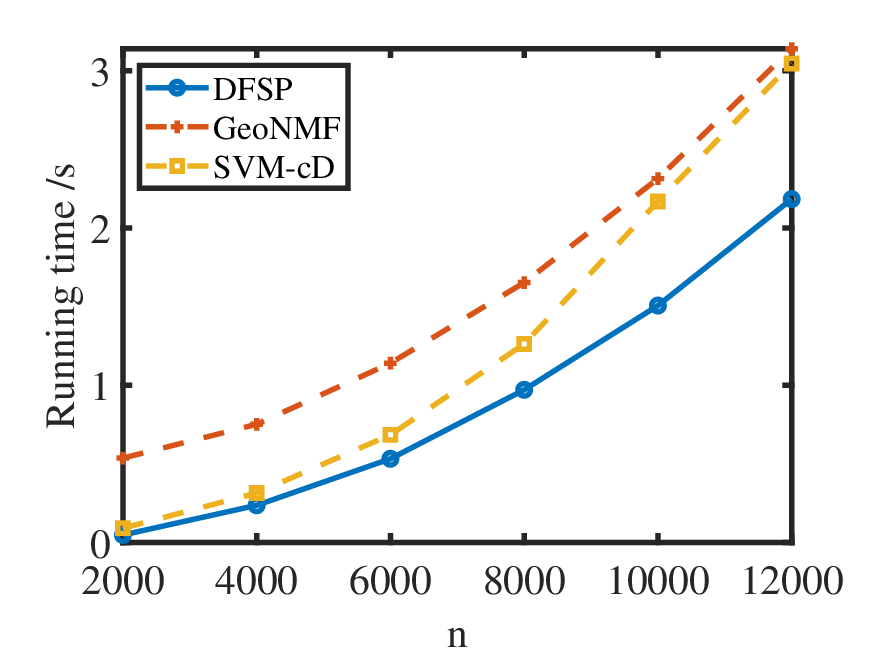}}
}
\resizebox{\columnwidth}{!}{
\subfigure[$\delta=0.1$]{\includegraphics[width=0.35\textwidth]{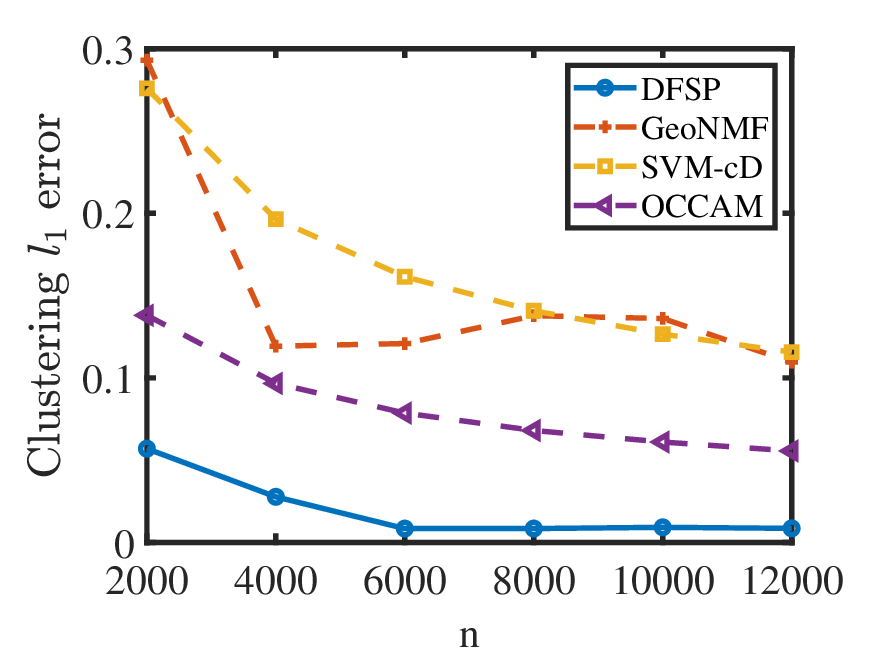}}
\subfigure[$\delta=0.1$]{\includegraphics[width=0.35\textwidth]{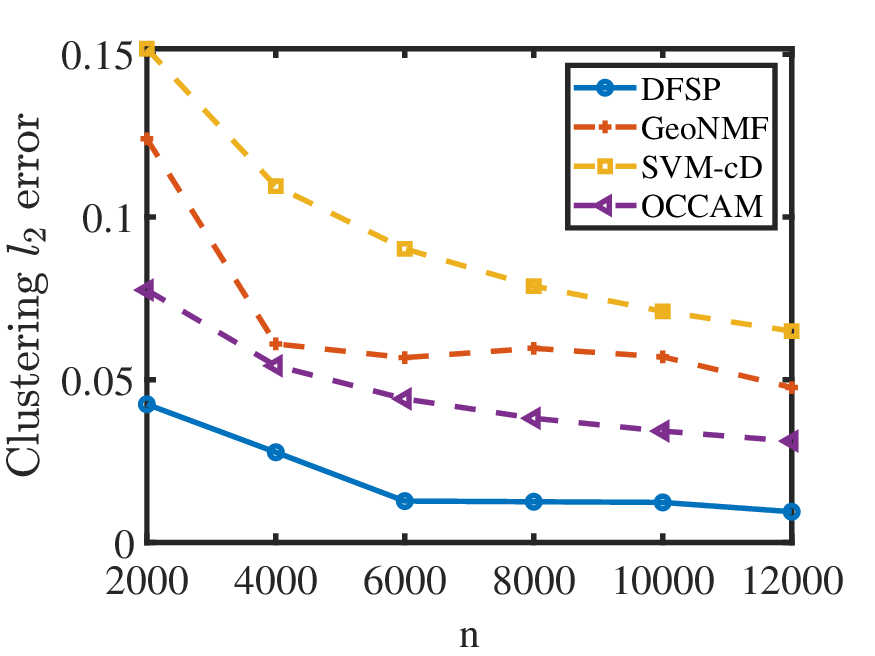}}
\subfigure[$\delta=0.1$]{\includegraphics[width=0.35\textwidth]{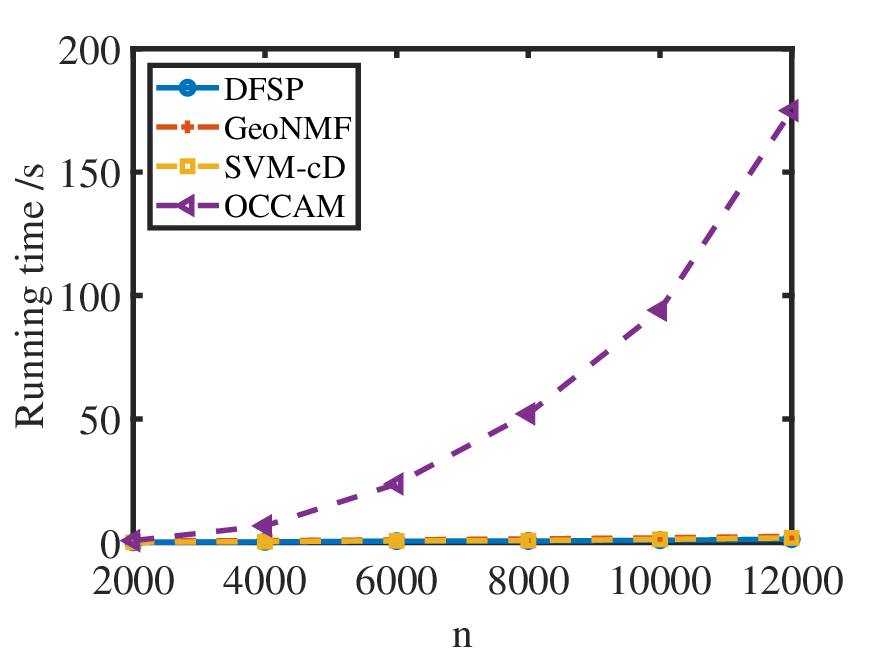}}
\subfigure[$\delta=0.1$]{\includegraphics[width=0.35\textwidth]{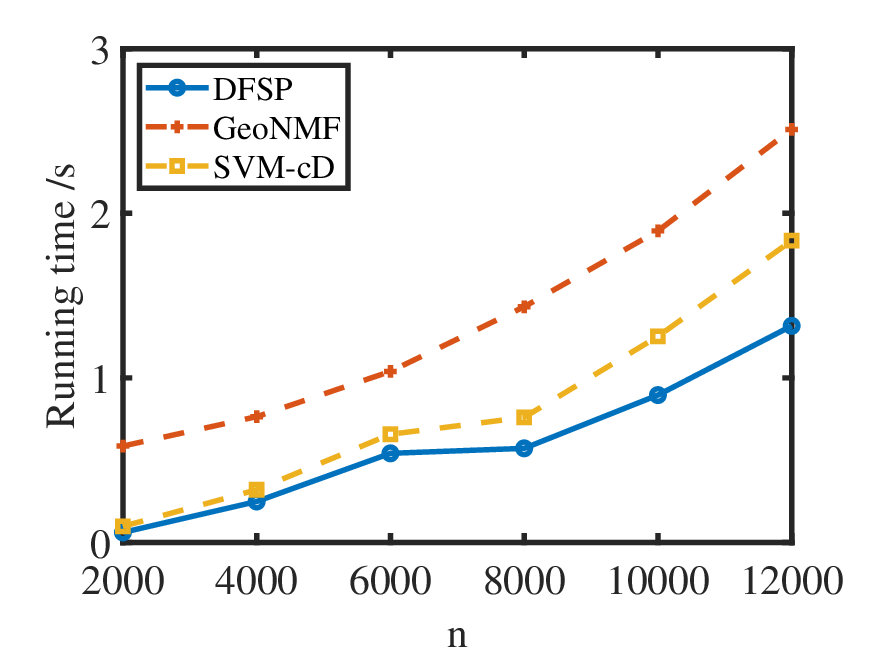}}
}
\caption{Numerical results of changing $n$.}
\label{EXTRASimN} 
\end{figure}
\end{appendices}


\bibliography{refMMDF}

\end{document}